\documentclass[runningheads]{llncs2}

\def\techreport{}

\usepackage[colorlinks=true]{hyperref}
\usepackage{graphicx}
\usepackage{subfigure}
\usepackage{amsmath,amssymb}
\usepackage{color}
\usepackage{multirow}
\usepackage{wrapfig}
\usepackage{ifthen}
\usepackage{float}
\usepackage{enumitem}
\usepackage{listings}
\setlist[itemize]{noitemsep, topsep=3pt}

\newcommand{\appref}[1]{Appendix~\ref{#1}}
\newcommand{\sectref}[1]{Section~\ref{#1}}

\newcommand{\figref}[1]{Fig.~\ref{#1}}

\newcommand{\eqnref}[1]{(\ref{#1})}
\newcommand{\thmref}[1]{Theorem~\ref{#1}}
\newcommand{\propref}[1]{Proposition~\ref{#1}}
\newcommand{\lemref}[1]{Lemma~\ref{#1}}
\newcommand{\defref}[1]{Definition~\ref{#1}}

\newcommand{\assumref}[1]{Assumption~\ref{#1}}

\newcommand{\figfigref}[2]{Figures~\ref{#1} and \ref{#2}}

\newtheorem{assumption}{Assumption}

\newcounter{exampcount}
\newcommand{\startpara}[1]{{%
\vskip6pt\noindent
{\bf #1.}}}

\def\ra{{\rightarrow}}

\def\cX{{\mathcal{X}}}

\def\Aset{\mathbb{A}}

\def\Nset{\mathbb{N}}
\def\Eset{\mathbb{E}}
\def\Pset{\mathbb{P}}

\def\Rset{\mathbb{R}}

\def\Tset{\mathbb{T}}

\def\ra{\rightarrow} %
\def\rmdef{\,\stackrel{\mbox{\rm {\tiny def}}}{=}}
\newcommand{\sem}[1]{ [ \! [ {#1}  ]  \! ]} %

\renewcommand{\leq}{\leqslant}
\renewcommand{\geq}{\geqslant}

\def\squareforqed{\hbox{\rlap{$\sqcap$}$\sqcup$}}
\def\qed{\ifmmode\squareforqed\else{\unskip\nobreak\hfil
\penalty50\hskip1em\null\nobreak\hfil\squareforqed
\parfillskip=0pt\finalhyphendemerits=0\endgraf}\fi}

\newcommand\game{{\sf G}}
\newcommand\pta{{\sf P}}
\newcommand\ptg{\mathsf{P}}

\newcommand\sinit{{\bar{s}}}

\newcommand\act{\mathit{A}}
\newcommand\Act{\act}

\newcommand\dist{{\mathit{Dist}}}

\newcommand\Prob{{\mathit{Prob}}}

\newcommand{\last}{\mathit{last}}
\newcommand{\ipaths}{\mathit{IPaths}}
\newcommand{\fpaths}{\mathit{FPaths}}

\newcommand{\coalition}[1]{\langle \! \langle {#1} \rangle \! \rangle}

\def\sat{{\,\models\,}}

\def\Act{{\mathit{Act}}}

\def\future{{{\mathtt F}\,}}

\newcommand{\ap}{\mathsf{a}}

\newcommand\probopP{{\mathtt P}}
\newcommand\probop[2]{\probopP_{#1}[\,{#2}\,]}

\newcommand{\ptgtuple}{( \Pi, \loc , \linit, \clocks, \Act , \langle \loc_i \rangle_{i \in \Pi}, \inv, \enb, \probt,\ptgrew)}

\newcommand{\clocks}{\mathcal{X}}

\newcommand{\npta}{{\sem{\pta}_\Nset}}
\newcommand{\tpta}{{\sem{\pta}_\Tset}}
\newcommand{\rptg}{{\sem{\ptg}_\Rset}}
\newcommand{\nptg}{{\sem{\ptg}_\Nset}}
\newcommand{\tptg}{{\sem{\ptg}_\Tset}}
\newcommand{\rptgc}{{\sem{\ptg}_\Rset^C}}
\newcommand{\nptgc}{{\sem{\ptg}_\Nset^C}}
\newcommand{\tptgc}{{\sem{\ptg}_\Tset^C}}

\newcommand{\loc}{\mathit{L}}
\newcommand{\probt}{\mathit{prob}}

\newcommand{\linit}{\overline{l}}
\newcommand{\inv}{\mathit{inv}}
\newcommand{\enb}{\mathit{enab}}

\newcommand{\digi}[2]{[ #1 ]_{#2}}
\newcommand{\dur}[1]{\mathit{dur}(#1)}

\newcommand{\CC}[1]{\mathit{CC}({#1})}

\newcommand{\ptarew}{{r}}
\newcommand{\ptgrew}{{r}}

\newcommand{\arew}{\ptarew_{\Act}} %
\newcommand{\lrew}{\ptarew_{L}} %

\usepackage{pgfplots}
\pgfplotsset{compat=1.13}

\usepackage{scalefnt}
\definecolor{blue}{rgb}{0,0,1}

\begin{document}

\title{Verification and Control of Turn-Based Probabilistic Real-Time Games}

\ifthenelse{\isundefined{\techreport}}{%
\author{Marta Kwiatkowska\inst{1}\orcidID{0000-0001-9022-7599} \and Gethin Norman\inst{2}\orcidID{0000-0001-9326-4344} \and David~Parker\inst{3}\orcidID{0000-0003-4137-8862}}
}{%
\author{Marta Kwiatkowska\inst{1} \and Gethin Norman\inst{2} \and David~Parker\inst{3}}
}

\institute{Department of Computing Science, University of Oxford, UK
\and
School of Computing Science, University of Glasgow, UK
\and 
School of Computer Science, University of Birmingham, UK}

\maketitle

\begin{abstract}
Quantitative verification techniques have been developed for the formal analysis
of a variety of probabilistic models, such as Markov chains, Markov decision process and their variants.
They can be used to produce guarantees on quantitative aspects of system behaviour,
for example safety, reliability and performance,
or to help synthesise controllers that ensure such guarantees are met.
We propose the model of turn-based probabilistic timed multi-player games,
which incorporates probabilistic choice, real-time clocks
and nondeterministic behaviour across multiple players.
Building on the digital clocks approach for the simpler model of probabilistic timed automata, we show how to compute the key measures that underlie quantitative verification,
namely the probability and expected cumulative price to reach a target.
We illustrate this on case studies from computer security and task scheduling.
\end{abstract}

\section{Introduction}

Probability is a crucial tool for modelling computerised systems.
We can use it to model uncertainty, for example
in the operating environment of an autonomous vehicle or a wireless sensor network,
and we can reason about systems that use randomisation,
from probabilistic routing in anonymity network protocols %
to symmetry breaking in communication protocols.

Formal verification of such systems can provide us with
rigorous guarantees on, for example,
the performance and reliability of computer networks~\cite{BHHK10},
the amount of inadvertent information leakage by a security protocol~\cite{ACPS12},
or the safety level of an airbag control system~\cite{AFG+09}.
To do so requires us to model and reason about a range of 
quantitative aspects of a system's behaviour:
probability, time, resource usage, and many others.

Quantitative verification techniques have been developed
for a wide range of probabilistic models.
The simplest are Markov chains,
which model the evolution of a stochastic system over discrete time.
Markov decision processes (MDPs) %
additionally include nondeterminism, which can be used either
to model the uncontrollable behaviour of an adversary or
to determine an optimal strategy (or policy) for controlling the system.
Generalising this model further still,
we can add a notion of real time,
yielding the model of probabilistic timed automata (PTAs).
This is done in the same way as for the widely used model of timed automata (TAs),
adding real-valued variables called clocks to the model.
Tools such as PRISM~\cite{KNP11} and Storm~\cite{DJKV17} support
verification of a wide range of properties of these different probabilistic models.

Another dimension that we can add to these models and verification techniques
is \emph{game-theoretic} aspects.
These can be used to represent, for example,
the interaction between an attacker and a defender in a computer security protocol~\cite{ACKP18}, between a controller and its environment, or
between participants in a communication protocol who have opposing goals~\cite{vdHW03}.
Stochastic multi-player games %
include choices made by multiple players
who can either collaborate or compete to achieve their goals.
Tool support for verification of stochastic games,
e.g.~PRISM-games~\cite{KPW18},
has also been developed and deployed successfully in a variety of application domains.

In this paper, we consider a modelling formalism that captures all these aspects:
probability, nondeterminism, time and multiple players.
We define a model called
\emph{turn-based probabilistic timed multi-player games} (TPTGs),
which can be seen as either an extension of PTAs to incorporate multiple players,
or a generalisation of stochastic multi-player games to include time.
Building on known techniques for the simpler classes of models,
we show how to compute key properties of these models,
namely probabilistic reachability
(the probability of reaching a set of target states)
and expected price reachability
(the expected price accumulated before reaching a set of target states).

Existing techniques for PTAs largely fall into two classes:
\emph{zone-based} and \emph{digital clocks},
both of which construct and analyse a finite-state abstraction of the model.
Zones are symbolic expressions representing sets of clock values.
Zone-based approaches for analysing PTAs were first introduced in~\cite{KNSS02,KNSW07,KNP09c} and recent work extended them to the analysis of expected time~\cite{JKNQ17} and expected price reachability~\cite{KNP17b}.
The digital clocks approach works by mapping real-valued clocks to integer-valued ones,
reducing the problem of solving a PTA to solving a (discrete-time) MDP.
This approach was developed for PTAs in~\cite{KNPS06} and recently extended to the analysis of partially observable PTAs in~\cite{NPZ17}. 

In this paper, we show how a similar idea can be used to reduce the verification problem for
TPTGs to an equivalent one over (discrete-time) stochastic games.
More precisely, for the latter, we use \emph{turn-based stochastic games} (TSGs).
We first present the model of TPTGs
and give two alternative semantics:
one using real-valued clocks and the other using (integer-valued) digital clocks.
Then, we prove the  correspondence between these two semantics.
Next, we demonstrate the application of this approach to two case studies
from the domains of computer security and task scheduling.
Using a translation from TPTGs to TSGs and the model checking tool PRISM-games,
we show that a variety of useful properties can be studied on these case studies.
\ifthenelse{\isundefined{\techreport}}{%
An extended version of this paper, with complete proofs, is available~\cite{extended}.
}{%
}

\startpara{Related Work}
Timed games were introduced and shown to be decidable in~\cite{MPS95,AMPS98,dAFH+03}. These games have since been extensively studied; we mention~\cite{TK99,CDF+05}, where efficient algorithms are investigated, and~\cite{BG14}, which concerns the synthesis of strategies that are robust to stochastic perturbation in clock values. Also related is the tool UPPAAL TIGA~\cite{BCD+07}, which allows the automated analysis of reachability and safety problems for timed games.

Priced (or weighted) timed games were introduced in~\cite{LMM02,AMM04,BCFL04}, which extend timed games by associating integer costs with locations and transitions, and optimal cost reachability was shown to be decidable under certain restrictions. The problem has since been shown to be undecidable for games with three or more clocks~\cite{BBR05,BBM06}. Priced timed games have recently been extended to allow partial observability~\cite{CDL+07} and to energy games~\cite{BMR+2018}. 

Two-player (concurrent) probabilistic timed games were introduced in~\cite{FKNT16}. The authors demonstrated that such games are not determined (even when all clock constraints are closed) and investigated the complexity of expected time reachability for such games. Stochastic timed games~\cite{BF09,BKK+10} are turn-based games where time delays are exponentially distributed. A similar model, based on interactive Markov chains~\cite{Her02}, is considered in \cite{BHK+12}.

\section{Background}

We start with some background and notation on \emph{turn-based stochastic games} (TSGs). For a set $X$, let $\dist(X)$ denote the set of discrete probability distributions over $X$ and $\Rset$ the set of non-negative real numbers.
\begin{definition}[Turn-based stochastic multi-player game]
A turn-based stochastic multi-player game (TSG) is a tuple $\game{=}(\Pi,S,\sinit,\act,\langle S_i \rangle_{i \in \Pi},\delta,R)$ where $\Pi$ is a finite set of players, $S$ is a (possibly infinite) set of states, $\sinit \in S$ is an initial state, $\act$ is a (possibly infinite) set of actions, $\langle S_i \rangle_{i \in \Pi}$ is a partition of the state space, $\delta : S{\times}\act \ra \dist(S)$ is a (partial) transition function and $R : S {\times} \act \ra \Rset$ is a price (or reward) function. 
\end{definition}
The transition function is partial in the sense that $\delta$ need not be defined for all state-action pairs. For each state $s$ of a TSG $\game$, there is a set of available actions given by $A(s) \rmdef \{ a \in \act \mid \delta(s,a) \; \mbox{is defined} \}$. The choice of which available action is taken in state $s$ is under the control of a single player: the player $i$ such that $s \in S_i$. If player $i$ selects action $a \in \act(s)$ in $s$, then the probability of transitioning to state $s'$ equals $\delta(s,a)(s')$ and a price of $R(s,a)$ is accumulated.

\startpara{Paths and strategies}
A path of a TSG $\game$ is a sequence $\pi = s_0 \xrightarrow{a_0} s_1 \xrightarrow{a_1} \cdots$ such that $s_i \in S$, $a_i \in A(s_i)$ and $\delta(s_i,a_i)(s_{i+1}){>}0$ for all $i {\geq} 0$. For a path $\pi$, we denote by $\pi(i)$ the $(i{+}1)$th state of the path, $\pi[i]$ the action associated with the $(i{+}1)$th transition and, if $\pi$ is finite, $\last(\pi)$ the final state. The length of a path $\pi$, denoted $|\pi|$, equals the number of transitions. For a path $\pi$ and $k {<} |\pi|$, let $\pi^{(k)}$ be the $k$th prefix of $\pi$. Let $\fpaths_\game$ and $\ipaths_\game$ equal the sets of finite and infinite paths starting in the initial state $\sinit$.

A strategy for player $i \in \Pi$ is a way of resolving the choice of action in each state under the control of player $i$, based on the game's execution so far. Formally, a strategy $\sigma$ for player $i \in \Pi$ is a function $\sigma_i : \{ \pi \in \fpaths_\game \mid \last(\pi) \in S_i \} \ra \dist(\act)$ such that, if $\sigma_i(\pi)(a){>}0$, then $a \in A(\last(\pi))$. The set of all strategies of player $i \in \Pi$ is represented by $\Sigma^i_\game$ (when clear from the context we will drop the subscript $\game$). A strategy for player $i$ is deterministic if it always selects actions with probability 1, and memoryless if it makes the same choice for any paths that end in the same state. 

A strategy profile for $\game$ takes the form $\sigma {=} \langle \sigma_i \rangle_{i \in \Pi}$, listing a strategy for each player. We use $\fpaths^\sigma$ and $\ipaths^\sigma$ for the sets of finite and infinite paths corresponding to the choices made by the profile $\sigma$ when starting in the initial state. For a given profile $\sigma$, the behaviour of $\game$ is fully probabilistic and we can define a probability measure $\Prob^\sigma$ over the set of infinite paths $\ipaths^\sigma$~\cite{KSK76}. 

\startpara{Properties}
Two fundamental properties of quantitative models are the probability of reaching a set of target states and the expected price accumulated before doing so. For a strategy profile $\sigma$ and set of target states $F$ of a TSG $\game$, the probability of reaching $F$ and expected price accumulated before reaching $F$ from the initial state $\sinit$ under the profile $\sigma$ are given by the following (again, when it is clear from the context, we will drop the subscript $\game$):
\[
\begin{array}{rcl}
\Pset^{\sigma}_\game(F) & \; \rmdef \; & \Prob^\sigma (\{ \pi \in \ipaths^\sigma \mid \pi(i) \in F \; \mbox{for some} \; i \in \Nset \}) \\
\Eset^{\sigma}_\game(F) & \; \rmdef \; & \int_{\pi \in \ipaths^\sigma} \mathit{rew}(\pi,F) \: \mathrm{d}\Prob^\sigma
\end{array}
\]
where for any infinite path $\pi$:
\[ \begin{array}{rcl}
\mathit{rew}(\pi,F) & \; \rmdef \; & \sum_{i=0}^{k_F} R(\pi(i),\pi[i])
\end{array} \]
and $k_F {=} \min \{ k{-}1 \mid \pi(k) \in F \}$ if $\pi(k) \in F$ for some $k \in \Nset$ and $k_F {=} \infty$ otherwise. 

To quantify the above properties over the strategies of the players, we consider a coalition $C \subseteq \Pi$ who try to maximise the property of interest, while the remaining players $\Pi {\setminus} C$ try to minimise it. Formally, we have the following definition:
\[
\begin{array}{rcl}
\Pset^C_\game(F) & \; \rmdef \; & \sup_{\sigma_1 \in \Sigma^1} \inf_{\sigma_2 \in \Sigma^2} \Pset^{\sigma_1,\sigma_2}_{\game^C}(F) \\
\Eset^C_\game(F) & \; \rmdef \; & \sup_{\sigma_1 \in \Sigma^1} \inf_{\sigma_2 \in \Sigma^2} \Eset^{\sigma_1,\sigma_2}_{\game^C}(F)
\end{array}
\]
where $\game^C$ is the two-player game constructed from $\game$ where the states controlled by player 1 equal $\cup_{i \in C} S_i$ and the states controlled by player 2 equal $\cup_{i \in \Pi {\setminus}C} S_i$.

The above definition yields the \emph{optimal value} of $\game$ if it is \emph{determined}, i.e., if the maximum value that the coalition $C$ can ensure equals the minimum value that the coalition $\Pi \setminus C$ can ensure. Formally, the definition of determinacy and optimal strategies for probabilistic reachability properties of TSGs are given below, and the case of expected reachability is analogous (replacing $\Pset$ with $\Eset$).
\begin{definition}
For a TSG $\game$, target $F$ and coalition of players $C$, we say the game $\game^C$ is \emph{determined} with respect to probabilistic reachability if:
\[ \begin{array}{rcl}
\sup_{\sigma_1 \in \Sigma^1} \inf_{\sigma_2 \in \Sigma^2} \Pset^{\sigma_1,\sigma_2}(F)  & \; = \; & \inf_{\sigma_2 \in \Sigma^2} \sup_{\sigma_1 \in \Sigma^1} \Pset^{\sigma_1,\sigma_2}(F) \, .
\end{array} \]
Furthermore, a strategy $\sigma_1^\star \in \Sigma_1$ is \emph{optimal} if\/
$\Pset^{\sigma_1^\star,\sigma_2}(F) \geq \Pset^C_\game(F)$ for all $\sigma_2 \in \Sigma^2$
and strategy $\sigma_2^\star \in \Sigma_2$ is \emph{optimal} if\/ $\Pset^{\sigma_1,\sigma_2^\star}(F) \leq \Pset^C_\game(F)$ for all $\sigma_1 \in \Sigma^1$.
\end{definition}
As we shall demonstrate, the games we consider are determined with respect to probabilistic and expected reachability, and optimal strategies exist. In particular, finite-state and finite-branching TSGs are determined~\cite{Krc09} and efficient techniques exist to approximate optimal values and optimal strategies~\cite{Con93,FV97}.
These techniques underlie the model checking algorithms for logics such as rPATL, %
defined for TSGs and implemented in the tool PRISM-games~\cite{KPW18}.

\section{Turn-based Probabilistic Timed Multi-Player Games}\label{tptgs-sect}

We now introduce \emph{turn-based probabilistic timed multi-player games} (TPTGs), a framework for modelling systems which allows probabilistic, non-deterministic, real-time and competitive behaviour. Let $\Tset \in \{ \Rset , \Nset \}$ be the time domain of either the non-negative reals or natural numbers.

\startpara{Clocks, valuations and clock constraints} We assume a finite set of \emph{clocks} $\clocks$. A \emph{clock valuation} is a function $v : \clocks \ra \Tset$; the set of all clock valuations is denoted $\Tset^\clocks$. Let $\mathbf{0}$ be the clock valuation that assigns the value 0 to all clocks. For any set of clocks $X \subseteq \clocks$ and clock valuation $v \in \Tset^\clocks$, let $v[X{:=}0]$ be the clock valuation such that, for any clock $x$, we have $v[X{:=}0](x)$ equals $0$ if $x \in X$ and $v(x)$ otherwise. Furthermore, for any time instant $t \in \Tset$, let $v{+}t$ be the clock valuation such that $(v{+}t)(x)=v(x){+}t$ for all $x \in \clocks$. A closed, diagonal-free clock constraint\footnote{A constraint is closed if does not contain strict inequalities and diagonal-free if there are no inequalities of the form $x {-} y \sim c$ for $x,y \in \clocks$, $\sim \in \{<,\leq,\geq,>\}$ and $c \in \Nset$.} $\zeta$ is a conjunction of inequalities of the form $x{\leq}c$ or $x{\geq}c$, where $x \in \clocks$ and $c \in \Nset$. We write $v \sat \zeta$ if the clock valuation $v$ satisfies the clock constraint $\zeta$ and use $\CC{\clocks}$ for the set of all clock constraints over $\clocks$.

We are now in a position to present the syntax and semantics of TPTGs.
\begin{definition}[TPTG syntax]\label{pta}
A \emph{turn-based probabilistic timed multi-player game} (TPTG) is a tuple $\ptg {=} \ptgtuple$ where:
\begin{itemize}
\item
$\Pi$ is a finite set of {\em players};
\item
$\loc$ is a finite set of {\em locations} and $\linit\in\loc$ is an {\em initial location};
\item
$\clocks$ is a finite set of {\em clocks;}
\item
$\Act$ is a finite set of {\em actions;}
\item
$\langle \loc_i \rangle_{i \in \Pi}$ is a partition of $\loc$;
\item
$\inv: \loc \ra \CC{\clocks}$ is an {\em invariant condition;}
\item
$\enb: \loc {\times} \act \ra \CC{\clocks}$ is an {\em enabling condition;}
\item
$\probt : \loc {\times} \act \ra \dist(2^{\clocks} {\times} \loc)$ is a (partial) {\em probabilistic transition function;}
\item
$\ptgrew = (\lrew,\arew)$ is a \emph{price structure} where $\lrew : L \ra \Nset$ is a location price function and $\arew : L {\times} \Act \ra \Nset$ an action price function.
\end{itemize}
\end{definition}
As for PTAs~\cite{NPS13}, a state of a TPTG $\ptg$ is a location-clock valuation pair $(l,v)$ such that the clock valuation satisfies the invariant $\inv(l)$. The transition choice in $(l,v)$ is under the control of the player $i$ where $l \in L_i$. A transition is a time-action pair $(t,a)$ which represents letting time $t$ elapse and then performing action $a$. Time can elapse if the invariant of the current location remains continuously satisfied and action $a$ can be performed only if the enabling condition is satisfied. If action $a$ is taken in location $l$, then the probability of moving to location $l'$ and resetting the set of clocks $X$ equals $\probt(l, a)(X, l')$. TPTGs have both location prices, which are accumulated at rate $\lrew(l)$ when time passes in location $l$, and action prices, where $\arew(l,a)$ is accumulated when performing action $a$ in location $l$. Formally, the semantics of a TPTG is a TSG defined as follows.
\begin{definition}[TPTG semantics]\label{ptg-sem}
For any time domain $\Tset \in \{ \Rset,\Nset \}$ and TPTG $\ptg {=} \ptgtuple$ the semantics of $\ptg$ with respect to the time domain $\Tset$ is the TSG $\tpta{=}(\Pi,S,\sinit,\act,\langle S_i \rangle_{i \in \Pi},\delta,R)$ where:
\begin{itemize}
\item $S = \{ (l,v) \in \loc {\times} \Tset^\clocks \mid v \sat \inv(l) \}$ and\/ $\sinit = (\linit,\mathbf{0});$
\item $\act = \Tset {\times} \Act$;
\item $S_i = \{ (l,v) \in S \mid l \in \loc_i \}$ for $i \in \Pi$;
\item for any $(l,v) \in S$ and $(t,a) \in \act$ we have $\delta((l,v),(t,a))=\mu$ if and only if $v{+}t' \sat \inv(l)$ for all $0 {\leq} t' {\leq} t$, $v{+}t \sat \enb(l,a)$ and for any $(l',v') \in S$:
\[ \begin{array}{rcl}
\mu(l',v') & = & \sum\limits_{X \subseteq \cX \wedge v'{=}(v{+}t)[X{:=}0]} \probt(l,a)(X,l')
\end{array} \]
\item $R((l,v),(t,a))= t{\cdot}\lrew(l) + \arew(l,a)$ for all $(l,v) \in S$ and $(t,a) \in \act$.
\end{itemize}
\end{definition}
We follow the approach of \cite{JKNT09,FKNT16,JKNQ17} and use time-action pairs in the transition function of \defref{ptg-sem}. As explained in~\cite{JKNQ17}, this yields a more expressive semantics than having separate time and action transitions.
\begin{figure}[t]
\centering
\hspace*{0.9cm}\includegraphics[scale=0.26]{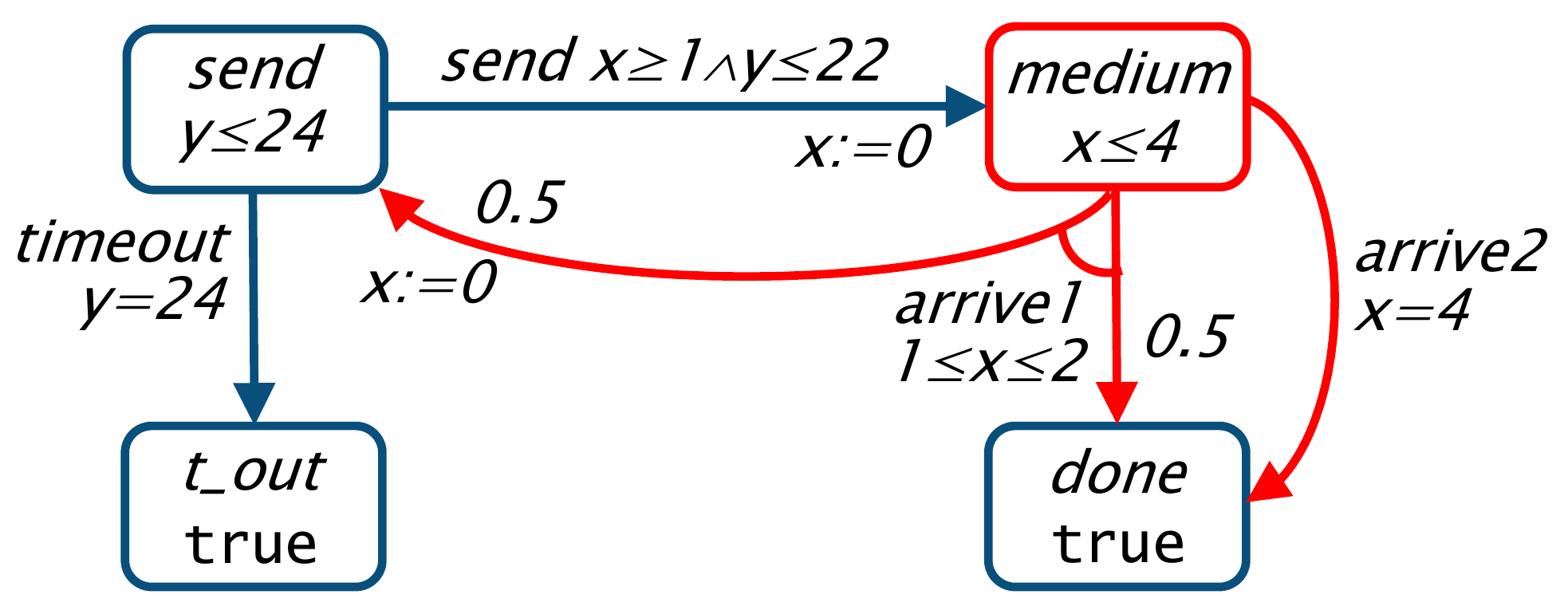}
\vspace*{-0.4cm}
\caption{An example TPTG}\label{examp-fig}
\vspace*{-0.5cm}
\end{figure}
\begin{example}
Consider the TPTG in \figref{examp-fig} which represents a simple communication protocol. There are two players: the sender and the medium, with the medium controlling the location $\mathit{medium}$ and the sender all other locations. The TPTG has two clocks: $x$ is used to keep track of the time it takes to send a message and $y$ the elapsed time. In the initial location $\mathit{send}$, the sender waits between 1 and 2 time units before sending the message. The message then passes through the medium that can either delay the message for between 1 and 2 time units after which it arrives with probability 0.5, or delay the message for 4 time units after which it arrives with probability 1. If the message does not arrive, then the sender tries to send it again until reaching a timeout after 24 time units.
\end{example}
As for PTAs, in the standard (dense-time) semantics for a TPTG the time domain $\Tset$ equals $\Rset$. This yields an infinite state model which is not amenable to verification. One approach that yields a finite state representation used in the case of PTAs is the digital clocks semantics~\cite{KNPS06}. This is based on replacing the real-valued clocks with clocks taking only values from a bounded set of integers. In order to give the definition for a TPTG $\ptg$, for any clock $x$ of $\ptg$ we define $k_x$ to be the greatest constant to which $x$ is compared in the clock constraints of $\ptg$. This allows us to use bounded clock values since, if the value of the clock $x$ exceeds $k_x$, then the exact value will not affect the satisfaction of the invariants and enabling conditions of $\ptg$, and therefore does not influence the behaviour.
\begin{definition}[Digital clocks semantics]
The digital clocks semantics of a TPTG $\ptg$, denoted $\npta$, is obtained from \defref{ptg-sem}, by setting $\Tset$ equal to $\Nset$ and for any $v \in \Nset^\clocks$,  $t \in \Nset$ and $x\in \clocks$ letting $(v{+}t)(x) = \min\{v(x){+}t, k_x{+}1 \}$.
\end{definition}
We restrict our attention to time-divergent (also called non-Zeno) behaviour. More precisely, we only consider strategies for the players that do not generate unrealisable executions, i.e., executions in which time does not advance beyond a certain point. We achieve this by restricting to TPTGs that satisfy the syntactic conditions for PTAs given in~\cite{NPS13}, derived from results on TAs~\cite{Tri99,TYB05}. In addition, we require the following assumptions to ensure the correctness of the digital clocks semantics.
\begin{assumption}\label{pta-assum}
For any TPTG $\ptg:$ (a) all invariants of $\pta$ are bounded;
(b) all clock constraints are closed and diagonal free;
(c) all probabilities are rational.
\end{assumption}
Regarding \assumref{pta-assum}(a), in fact bounded TAs are as expressive as standard TAs~\cite{Behrmann:2001}, and this result carries over to TPTGs.

To facilitate higher-level modelling, PTAs can be extended with parallel composition, discrete variables, urgent transitions and locations and resetting clocks to integer values~\cite{NPS13}. We can extend TPTGs in a similar way, and will use these constructs in \sectref{case-sect}.

\section{Correctness of the Digital Clocks Semantics}\label{correct-sect}

We now show that, under \assumref{pta-assum}, optimal probabilistic and expected price reachability values agree under the digital and dense-time semantics. As for PTAs~\cite{NPS13}, by modifying the TPTG under study, we can reduce time-bounded probabilistic reachability properties to probabilistic reachability properties and both expected time-bounded cumulative price properties and expected time-instant price properties to expected reachability properties. In each case the modifications to the TPTG preserve \assumref{pta-assum}, and therefore the digital clocks semantics can also be used to verify these classes of properties.

For the remainder of this section, we fix a TPTG $\ptg$, coalition of players $C$ and set of target locations $F \subseteq \loc$, and let $F_\Tset = \{ (l,v) \in F {\times} \Tset^\clocks \mid v \sat \inv(l) \}$ for $\Tset \in \{ \Rset,\Nset\}$. We have omitted the proofs that closely follow those for PTAs~\cite{KNPS06}. 
\ifthenelse{\isundefined{\techreport}}{%
The missing proofs can be found in~\cite{extended}. 
}{%
The missing definitions and proofs can be found in \appref{appendix}. 
}

We first present results relating to the determinacy and existence of optimal strategies for the games $\rptgc$ and $\nptgc$ and a correspondence between the strategy profiles of $\nptgc$ and $\rptgc$.
\begin{proposition}\label{determined-prop}
For any TPTG $\ptg$ satisfying \assumref{pta-assum}, the games $\rptgc$ and $\nptgc$ are determined and have optimal strategies for both probabilistic and expected price reachability properties.
\end{proposition}
\begin{proof}
In the case of $\rptgc$, the result follows from \cite{FKNT10a} and \assumref{pta-assum}, i.e. since all clock constraints are closed. Considering $\nptgc$, the result follows from the fact that the game has a finite state space and is finitely branching~\cite{Krc09}. \qed
\end{proof}
\begin{proposition}\label{strat2-prop} 
For any strategy profile $\sigma'$ of $\nptgc$, there exists a strategy profile $\sigma$ of $\rptgc$
such that 
$\Pset^{\sigma}(F_\Rset) = \Pset^{\sigma'}(F_\Nset)$ and $\Eset^{\sigma}(F_\Rset) = \Eset^{\sigma'}(F_\Nset)$.
\end{proposition}
Using the $\epsilon$-digitization approach of   TAs~\cite{HMP92}, which has been extended to PTAs in~\cite{KNPS06}, the following theorem follows, demonstrating the correctness of the digital clocks semantics for probabilistic reachability properties.
\begin{theorem}\label{reach-thm}
For any TPTG $\ptg$ satisfying \assumref{pta-assum}, coalition of players $C$ and set of locations $F \subseteq \loc:$ $\Pset^{C}_\rptg (F_{\Rset}) = \Pset^{C}_\nptg (F_{\Nset})$.
\end{theorem}
For expected price reachability properties, we extend the approach of \cite{KNPS06}, by first showing that, for any fixed (dense-time) profile $\sigma$ of $\rptgc$ and $n \in \Nset$, there exist profiles of $\nptgc$ whose expected price 
of reaching the target locations $F$ within $n$ transitions 
that provide lower and upper bounds for that of $\sigma$.  

For $\Tset \in \{ \Rset,\Nset\}$, profile $\sigma$ of $\tptgc$, and finite path $\pi \in \fpaths^\sigma$ we inductively define the values
$\langle \Eset^\sigma_n(\pi,F_\Tset) \rangle_{n \in \Nset}$ which equal the expected price, under the profile $\sigma$, of reaching the target $F_\Tset$ after initially performing the path $\pi$ within $n$ steps. To ease presentation we only give the definition for deterministic profiles.
\begin{definition}\label{EB-def}
For $\Tset \in \{ \Rset,\Nset\}$, strategy profile $\sigma{=}(\sigma_1,\sigma_2)$ of $\tptgc$ and  
finite path $\pi$ of the profile let $\Eset^{\sigma_1,\sigma_2}_{0} (\pi,F_{\Tset})=0$ and for any $n \in \Nset$, if $\mathit{last}(\pi){=}(l,v) \in S_i$ for $1{\leq}i{\leq}2$, $\sigma_i(\pi){=}(t,a)$ and $\mu = P_\tptg((l,v),(t,a))$, then:
\begin{align*}
\Eset^{\sigma}_{n+1} (\pi,F_{\Tset}) =  \left\{ \begin{array}{cl}
0 & \; \mbox{if $(l,v) \in F_{\Tset}$} \\
\lrew(l){\cdot}t + \arew(l,a) + \sum\limits_{s' \in S} \mu(s') 
\cdot \Eset^{\sigma}_{n} (\pi \xrightarrow{t,a}s',F_{\Tset}) & \;\mbox{otherwise.}
\end{array} \right.
\end{align*}
\end{definition}
We require the following properties of these expected price reachability properties. These then allow us to prove the correctness of the digital clocks semantics for expected price reachability properties (\thmref{exp2-thm} below).
\begin{lemma}\label{lim-lem}
For $\Tset \in \{ \Rset,\Nset\}$ and profile $\sigma$ of $\tptgc$, the sequence
$\langle \Eset^{\sigma}_{n} (F_{\Tset}) \rangle_{n \in \Nset}$ is non-decreasing and converges to $\Eset^{\sigma}(F_{\Tset})$, and, for any player 1 strategy $\sigma_1$ of $\tptgc$, the sequence of functions $\Eset^{\sigma_1,\cdot}_{n} (F_{\Tset}) : \Sigma^2 \ra \Rset$ converges uniformly.
Furthermore, for any player 1 strategy $\sigma_1$, the sequence $\langle  \inf\nolimits_{\sigma_2 \in \Sigma^2} \Eset^{\sigma_1,\sigma_2}_{n} (F_{\Tset}) \rangle_{n \in \Nset}$ is non-decreasing and converges to $\inf\nolimits_{\sigma_2 \in \Sigma^2}\Eset^{\sigma_1,\sigma_2}(F_{\Tset})$, and the sequence of functions $\inf\nolimits_{\sigma_2 \in \Sigma^2} \Eset^{\cdot,\sigma_2}_{n} (F_{\Tset}) : \Sigma^1 \ra \Rset$ converges uniformly.
\end{lemma}
\begin{proof}
In each case, proving that the sequence is non-decreasing and converges follows from \defref{EB-def}. Uniform convergence follows from showing the set of strategies for players is compact and using the fact that the sequences are non-decreasing and converge pointwise \cite[Theorem~7.13]{Rud76}. In the case when $\Tset{=}\Nset$, compactness follows from the fact the action set is finite, while for $\Tset{=}\Rset$ we must restrict to PTAs for which all invariants are bounded (\assumref{pta-assum}) to ensure the action set is compact. \qed
\end{proof}
\begin{lemma}\label{variable-lem}
For any strategy profile $\sigma$ of $\rptgc$
and $n \in \Nset$, there exist strategy profiles $\sigma^{\mathit{lb}}$ and $\sigma^{\mathit{ub}}$ of $\nptgc$ such that: $\Eset^{\sigma^{\mathit{lb}}}_{n}(F_{\Nset}) 
\; \leq \; 
\Eset^{\sigma}_{n}(F_{\Rset})
\; \leq \;
\Eset^{\sigma^{\mathit{ub}}}_{n}(F_{\Nset})$.
\end{lemma}
\begin{theorem}\label{exp2-thm}
For any TPTG $\ptg$ satisfying \assumref{pta-assum}, coalition of players $C$ and set of locations $F \subseteq \loc:$ $\Eset_\rptg^C (F_{\Rset})  = \Eset_\nptg^C (F_{\Nset} )$.
\end{theorem}
\begin{proof}
Consider any $n \in \Nset$.
Using \lemref{variable-lem} it follows that, for any profile $\sigma{=}(\sigma_1,\sigma_2)$ of $\rptg$, there exist profiles $\sigma^{\mathit{lb}}{=}(\sigma_1^{\mathit{lb}},\sigma_1^{\mathit{lb}})$ and $\sigma^{\mathit{ub}}{=}(\sigma_1^{\mathit{ub}},\sigma_1^{\mathit{ub}})$ of $\nptgc$ such that:
\[
\Eset^{\sigma_1^{\mathit{lb}},\sigma_2^{\mathit{lb}}}_{n}(F_{\Nset})
\; \leq \;
\Eset^{\sigma_1,\sigma_2}_{n}(F_{\Rset}) 
\; \leq \;
\Eset^{\sigma_1^{\mathit{ub}},\sigma_2^{\mathit{ub}}}_{n}(F_{\Nset}) \, .
\]
On the other hand, using the construction in the proof of \propref{strat2-prop}, for any profile $\sigma'{=}(\sigma_1',\sigma_2')$ of $\nptg$, there exists a profile $\sigma{=}(\sigma_1,\sigma_1)$ of $\rptgc$ such that:
\[
\Eset^{\sigma_1,\sigma_2}_n(F_\Nset) \; = \;  \Eset^{\sigma_1',\sigma_2'}_n(F_\Rset) \, .
\]
Combining these results with \propref{determined-prop} it follows that:
\[  \begin{array}{c}
\sup\nolimits_{\sigma_1'  \in \Sigma^1_\nptgc} \!\! \inf\nolimits_{\sigma_2' \in \Sigma^2_\nptgc} \Eset^{\sigma_1',\sigma_2'}_n(F_{\Nset})
\; = \; 
\sup\nolimits_{\sigma_1  \in \Sigma^1_\rptgc} \!\! \inf\nolimits_{\sigma_2 \in \Sigma^2_\rptgc} \Eset^{\sigma_1,\sigma_2}_n(F_{\Rset}) \, .
\end{array} \]
Since $n \in \Nset$ was arbitrary, we have:
\[
\begin{array}{c}
\lim\limits_{n \ra \infty} \sup\nolimits_{\sigma_1'  \in \Sigma^1_\nptgc} \!\! \inf\nolimits_{\sigma_2' \in \Sigma^2_\nptgc} \Eset^{\sigma_1',\sigma_2'}_n(F_{\Nset})
= 
\lim\limits_{n \ra \infty} \sup\nolimits_{\sigma_1  \in \Sigma^1_\rptgc} \!\! \inf\nolimits_{\sigma_2 \in \Sigma^2_\rptgc} \Eset^{\sigma_1,\sigma_2}_n(F_{\Rset})
\end{array}
\]
and hence using \lemref{lim-lem} it follows that:
\[ \begin{array}{c}
\sup\nolimits_{\sigma_1'  \in \Sigma^1_\nptgc} \!\! \inf\nolimits_{\sigma_2' \in \Sigma^2_\nptgc} \Eset^{\sigma_1',\sigma_2'}(F_{\Nset})
\; = \;
\sup\nolimits_{\sigma_1  \in \Sigma^1_\rptgc} \!\! \inf\nolimits_{\sigma_2 \in \Sigma^2_\rptgc} \Eset^{\sigma_1,\sigma_2}(F_{\Rset}) \, .
\end{array} \]
The fact that the limit can move inside the $\sup$ and $\inf$ operators on both sides of the inequality follows from the uniform convergence results of \lemref{lim-lem}. \qed
\end{proof}

\section{Case Studies}\label{case-sect}

In this section, we apply our approach to two case studies, a security protocol and a scheduling problem, both of which have been previously modelled as PTAs~\cite{KNPS06}. In both case studies, by working with games we are able to give more realistic models that overcome the limitations of the earlier PTA models. We specify the finite-state TSG digital clocks semantic models of the case studies using the PRISM language and employ the PRISM-games tool~\cite{KPW18} to perform the analysis. 
Using PRISM-games we are not only able to find optimal probabilistic and expected reachability values, but also synthesise optimal strategies for the players. PRISM files for the case studies are available from~\cite{files}.
%

\startpara{Non-repudiation protocol} Markowitch and Roggeman's non-repudiation protocol for information transfer~\cite{MR99} is designed to allow an originator $O$ to transfer information to a recipient $R$ while guaranteeing non-repudiation, that is, neither $O$ nor $R$ can deny that they participated in the transfer. 

Randomisation is fundamental to the protocol as, in the initialisation step, $O$ randomly selects a positive integer $N$ that is never revealed to $R$ during execution. Timing is also fundamental as, to prevent $R$ potentially gaining an advantage, if $O$ does not receive an acknowledgement within a specific timeout value (denoted $\mathit{AD}$), the protocol is stopped and $O$ states $R$ is trying to cheat. In previous PTA models of the protocol~\cite{LMT05,NPS13} the originator $O$ had fixed behaviour, while the choices of a (malicious) recipient $R$ included varying the delay between receiving a message from $O$ and sending an acknowledgement. By modelling the protocol as a two-player game we can allow both $O$ and $R$ to make choices which can depend on the history, i.e., the previous behaviour of the parties. The game is naturally turn-based since, in each round, first $O$ sends a message after a delay of their choosing and, after receiving this message, $R$ can respond with an acknowledgement after a delay of their choosing.

\begin{figure}[t]
\centering
{\subfigure[Originator]{\includegraphics[scale=0.26]{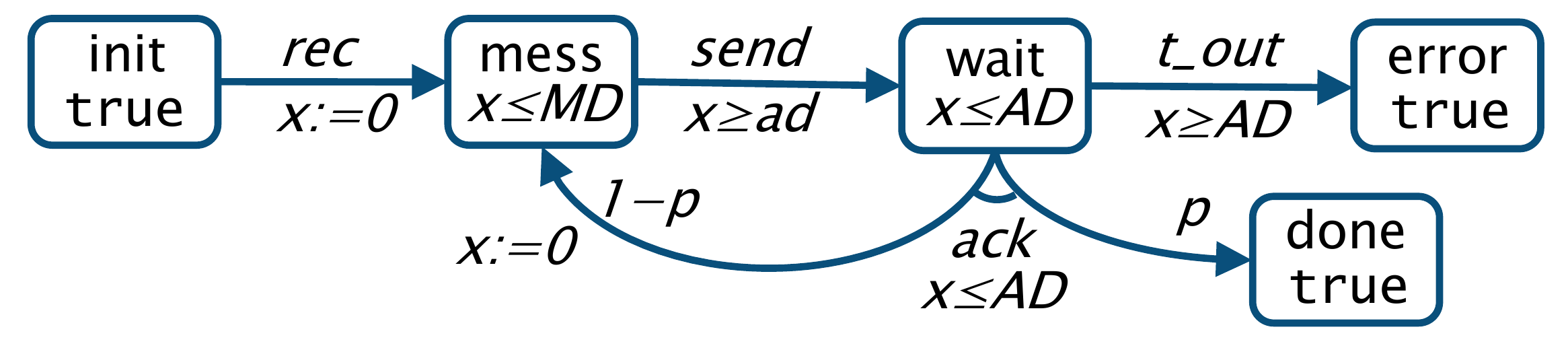}}}
\hfil
{\subfigure[Honest recipient]{\includegraphics[scale=0.26]{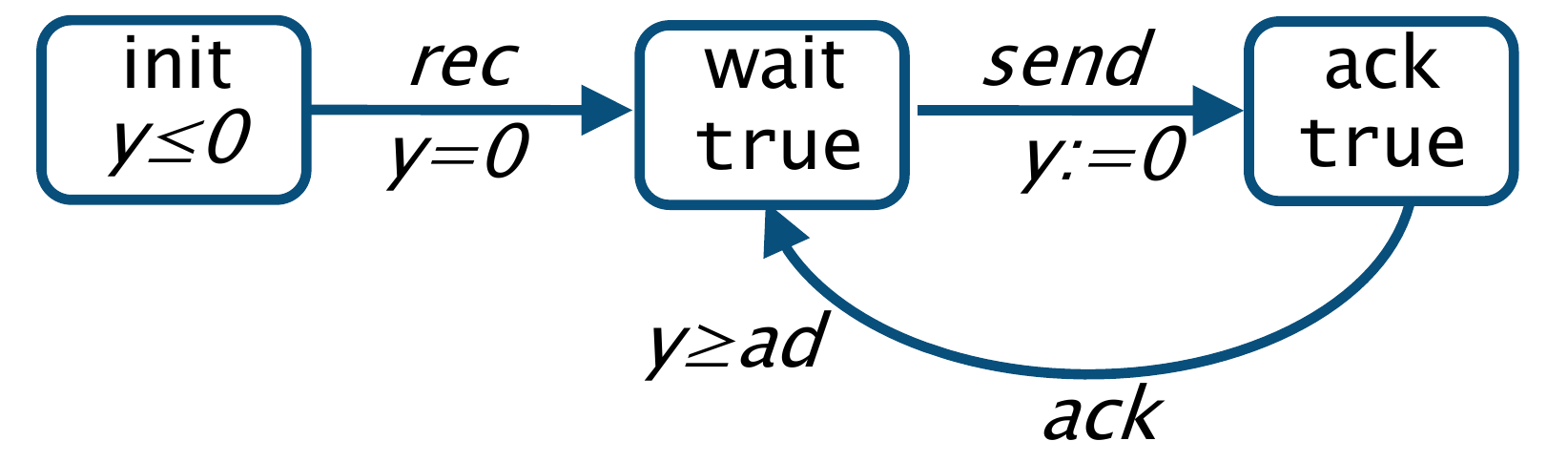}}}
\vspace*{-0.2cm}
\caption{PTAs used to model the non-repudiation protocol}\label{repudiation-fig}
\vspace*{-0.6cm}
\end{figure}

We first consider an `honest' version of the protocol where both $O$ and $R$ can choose delays for their messages but do follow the protocol (i.e., send messages and acknowledgements before timeouts occur). The component PTA models for $O$ and $R$ are presented in \figref{repudiation-fig}. In the PTA for $O$, the message delay is between $\mathit{md}{=}2$ and $\mathit{MD}{=}9$ time units, while the acknowledgement delay is at least $\mathit{ad}{=}1$ time units and $\mathit{AD}{=}5$ is the timeout value.
In addition, the probabilistic choice of $N$ is made using a geometric distribution with parameter $p\in(0,1]$. The parallel composition of these two components then gives the TPTG model of the protocol by assigning control of locations to either $O$ or $R$, based on which party decides on the delay. There is a complication in the location where $O$ is waiting and $R$ is sending an acknowledgement, as the delay before sending the acknowledgement controlled by $R$, while if the timeout is reached $O$ should end the protocol. However, since $O$'s behaviour is deterministic in this location, we assign this location to be under the control of $R$, but add the constraints that an acknowledgement can only be sent before the timeout is reached. If $O$'s behaviour was not deterministic, then a turn-based model would not be sufficient and the protocol would need to be modelled as a concurrent game.

We also consider two `malicious' versions of the protocol, one in which $R$ is allowed to guess which is the last message (malicious version 1) and a version further extended by giving $R$ additional power through a probabilistic decoder that can decode a message with probability $0.25$ before $O$ will timeout (malicious version 2). The TPTG models follow the same structure as that for the `honest' version, requiring that, in the locations where $O$ is waiting for an acknowledgement, once the timeout has been reached the only possible behaviour is for the protocol to terminate and $O$ states $R$ is trying to cheat.

For the `honest' version, \figfigref{honest-time-fig}{honest-expected-fig} present results when different coalitions try to maximise the probability the protocol terminates successfully by time $T$ when $p{=}0.01$ and $p{=}0.1$ and minimise expected time for successful termination as the parameter $p$ varies. More precisely, we consider the coalition of both players ($\coalition{O,R}$), a single player ($\coalition{O}$ or $\coalition{R}$) and the empty coalition ($\coalition{}$). Using a PTA model, only the first and last cases could be considered. As can be seen, both parties have some control over the time it takes for the protocol to complete and $O$ has greater power as it can delay messages longer than $R$ (if $R$ delays too long then $O$ will terminate the protocol stating $R$ is cheating).

\begin{figure}[t]
\centering
\includegraphics[scale=0.25]{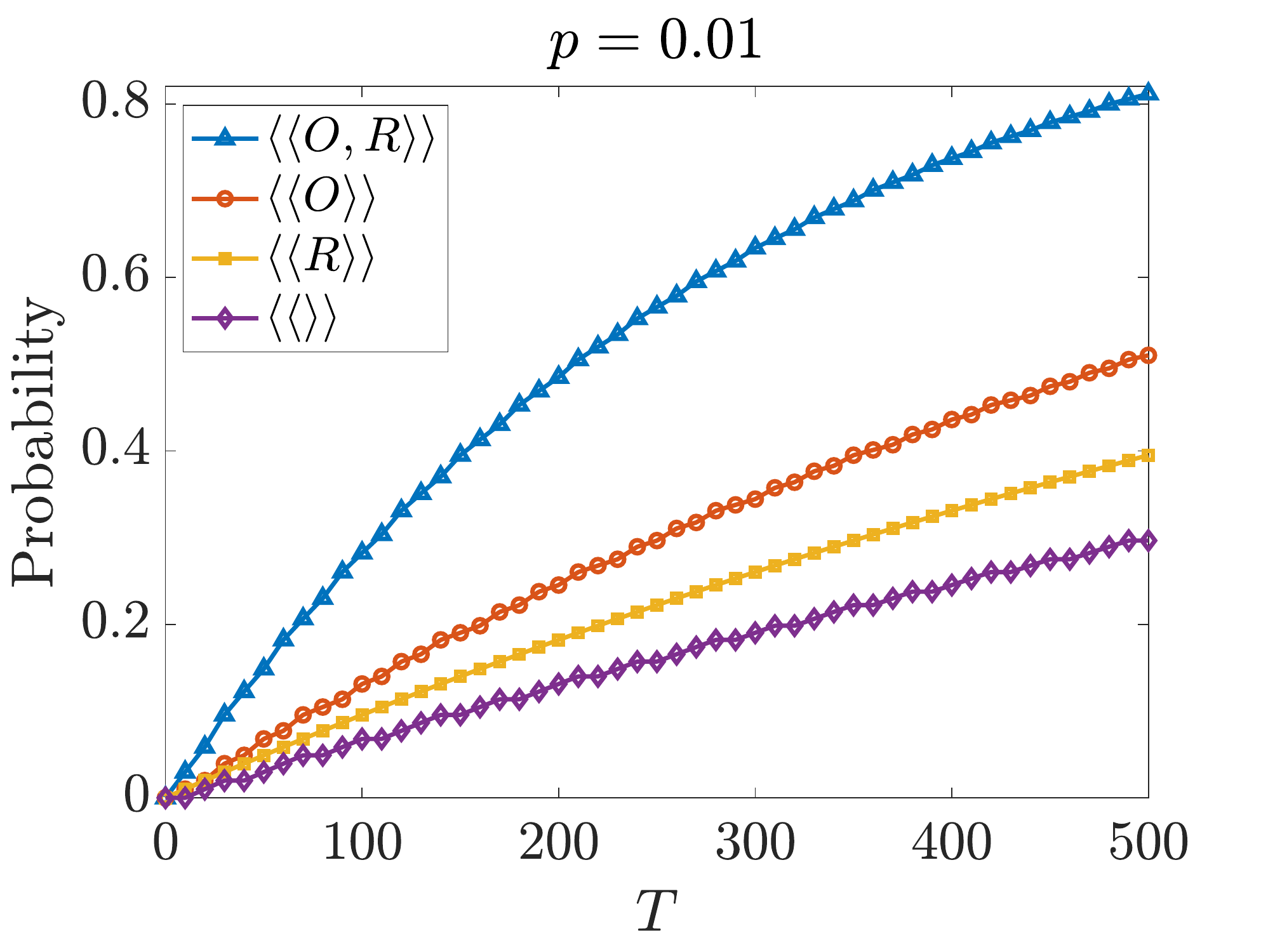}
\hfil
\includegraphics[scale=0.25]{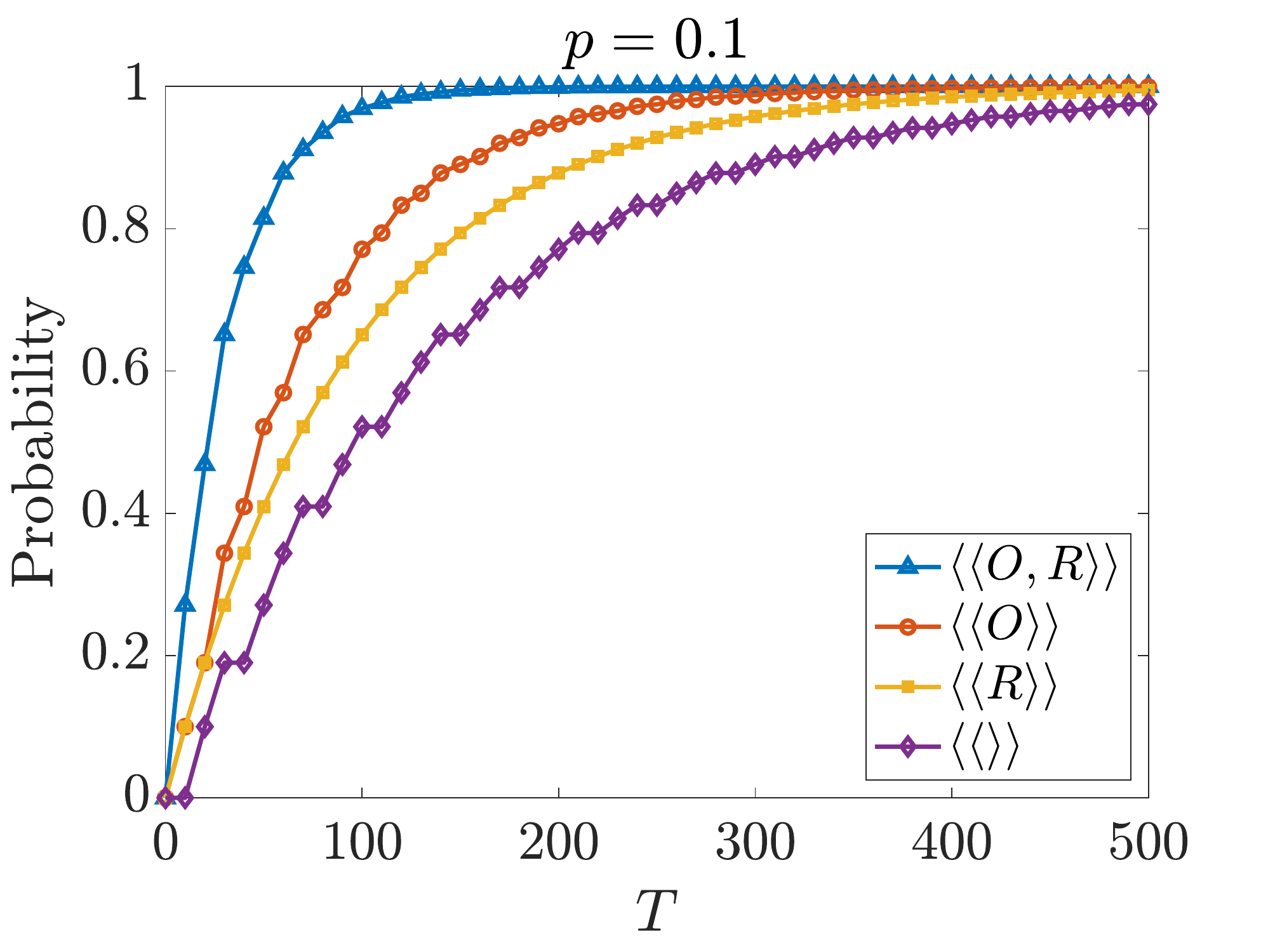}
\vspace*{-0.4cm}
\caption{Max. probability the protocol terminates successfully by time $T$ (honest version)}\label{honest-time-fig}
\vspace*{-0.4cm}
\end{figure}
\begin{figure}[t]
\centering
\includegraphics[scale=0.25]{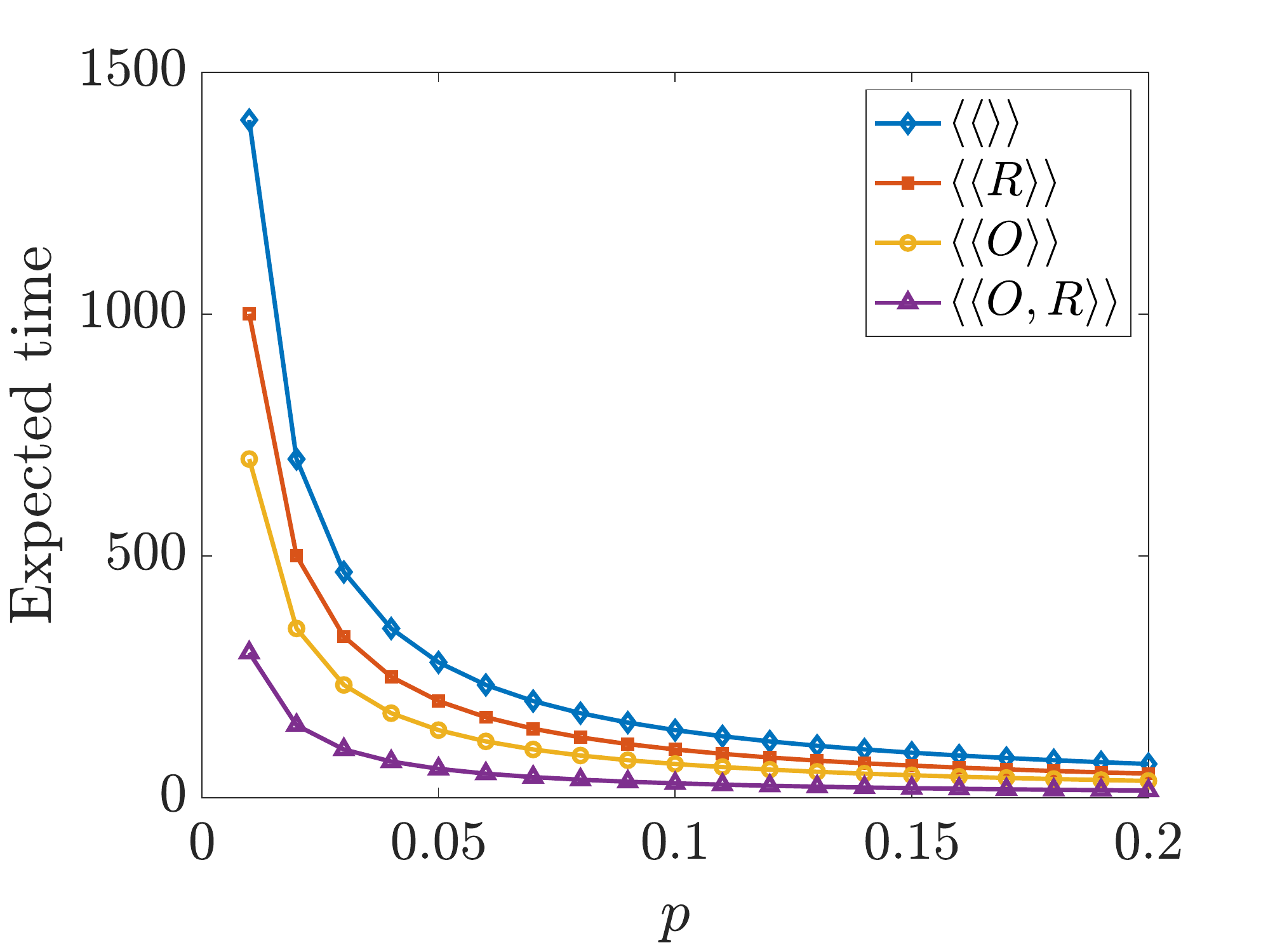}
\hfil
\includegraphics[scale=0.25]{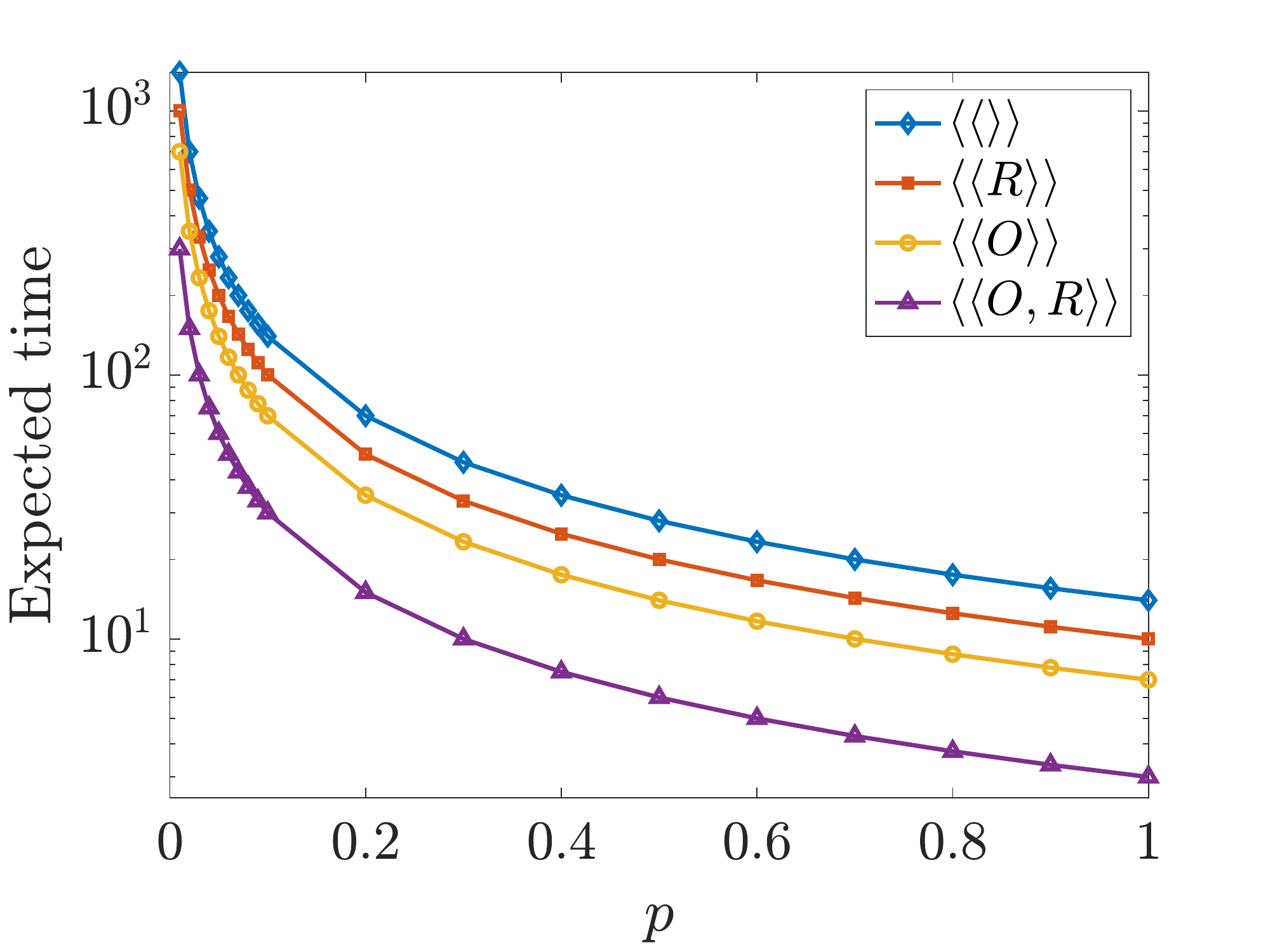}
\vspace*{-0.4cm}
\caption{Min. expected time until the protocol terminates successfully (honest version)}\label{honest-expected-fig}
\vspace*{-0.6cm}
\end{figure}

In the case of the versions with a malicious recipient, in \figfigref{basic-time-fig}{malicious-time-fig} we have plotted the maximum probability the recipient gains information by time $T$ for versions 1 and 2 respectively. We have included the cases where $O$ works against $R$ ($\coalition{R}$) and where they collaborate ($\coalition{O,R}$). As we can see, although $O$ cannot reduce the probability of $R$ obtaining information, it can to some extent increase the time it takes $R$ to obtain this information.

\begin{figure}[t]
\centering
\includegraphics[scale=0.25]{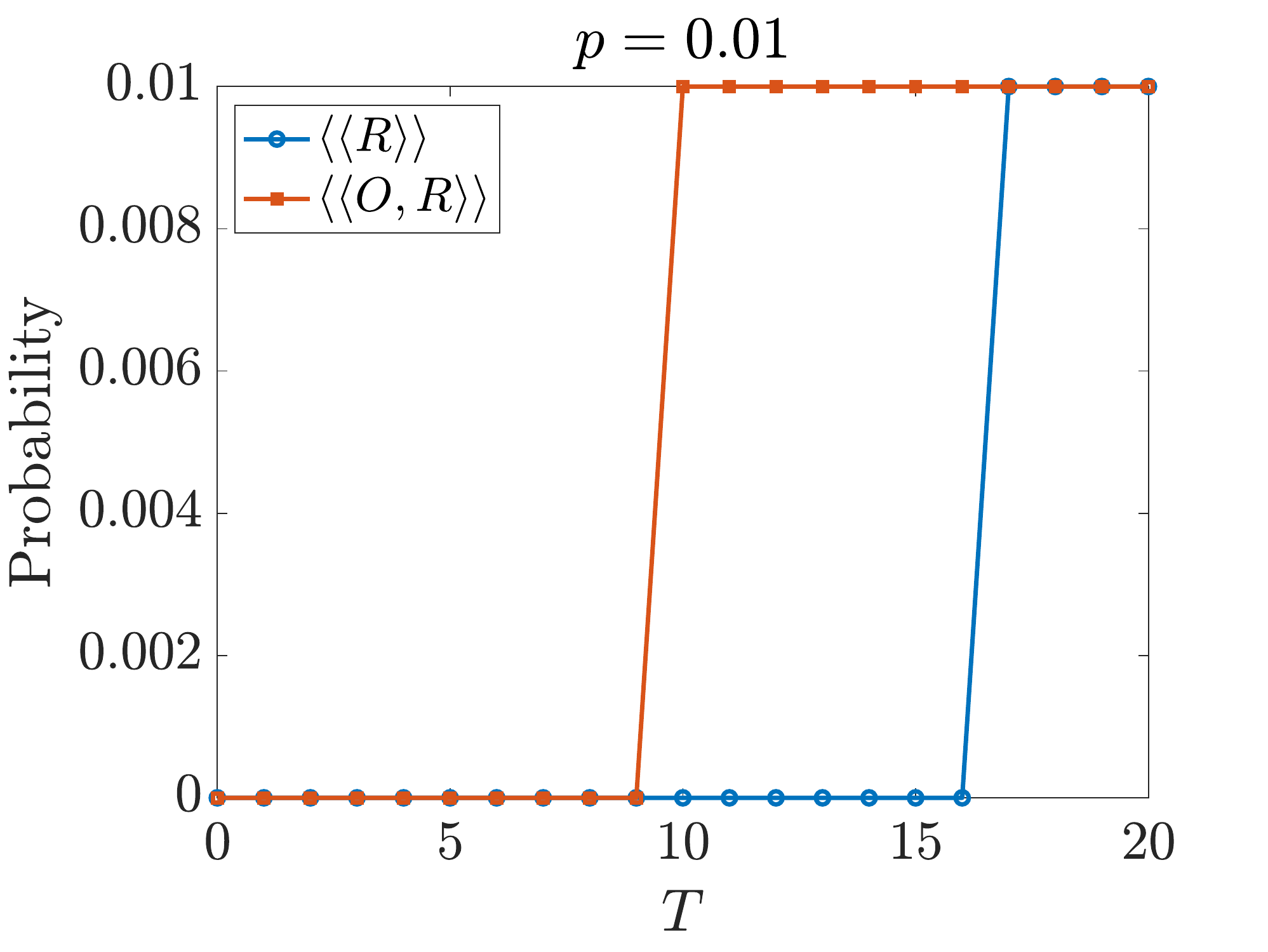}
\hfil
\includegraphics[scale=0.25]{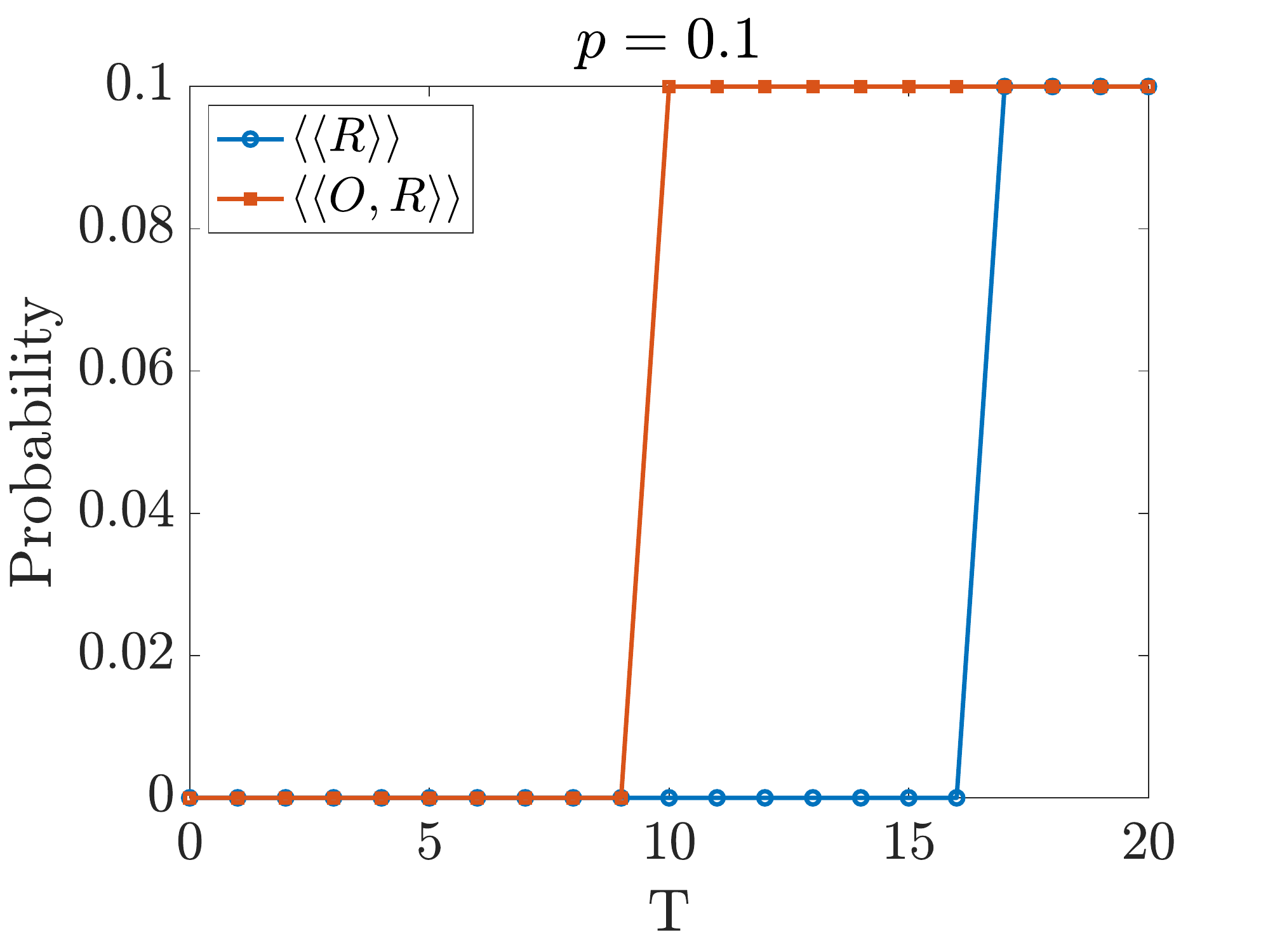}
\vspace*{-0.4cm}
\caption{Maximum probability $R$ gains information by time $T$ (malicious version 1)}\label{basic-time-fig}
\vspace*{-0.2cm}
\end{figure}
\begin{figure}[t]
\centering
\includegraphics[scale=0.25]{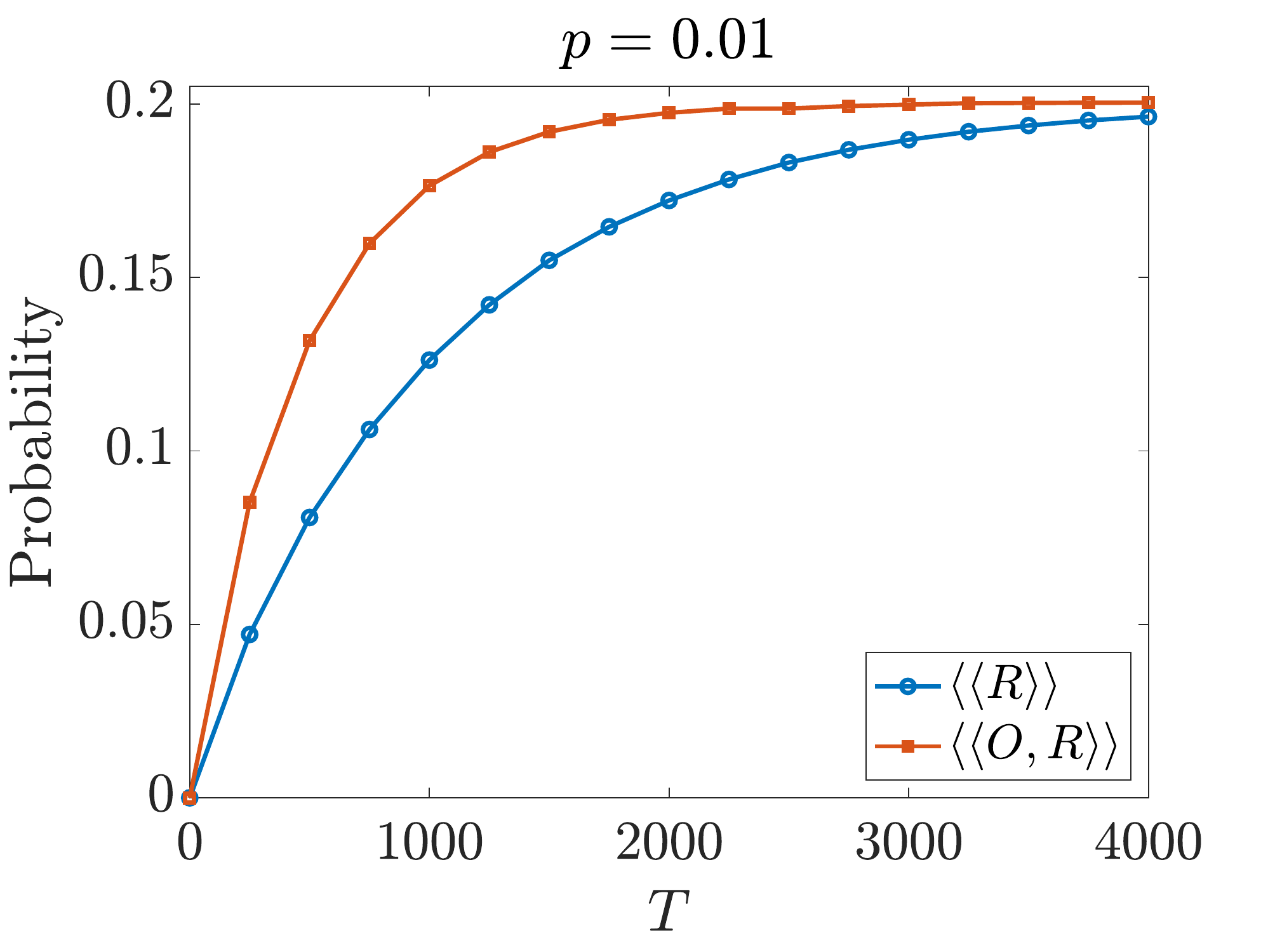}
\hfil
\includegraphics[scale=0.25]{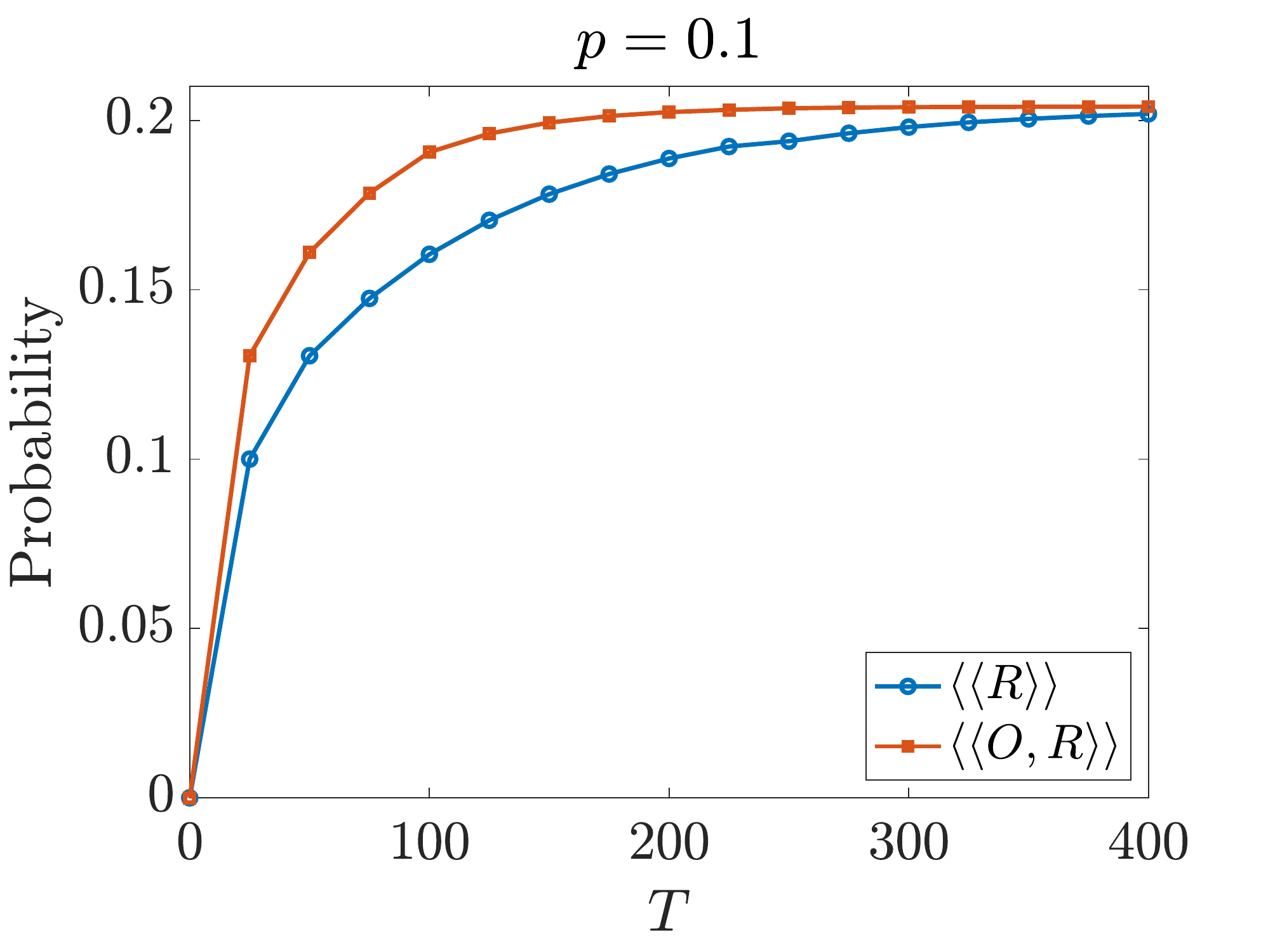}
\vspace*{-0.4cm}
\caption{Maximum probability $R$ gains information by time $T$ (malicious version 2)}\label{malicious-time-fig}
\vspace*{-0.6cm}
\end{figure}
\begin{figure}[t]
\centering
\includegraphics[scale=0.26]{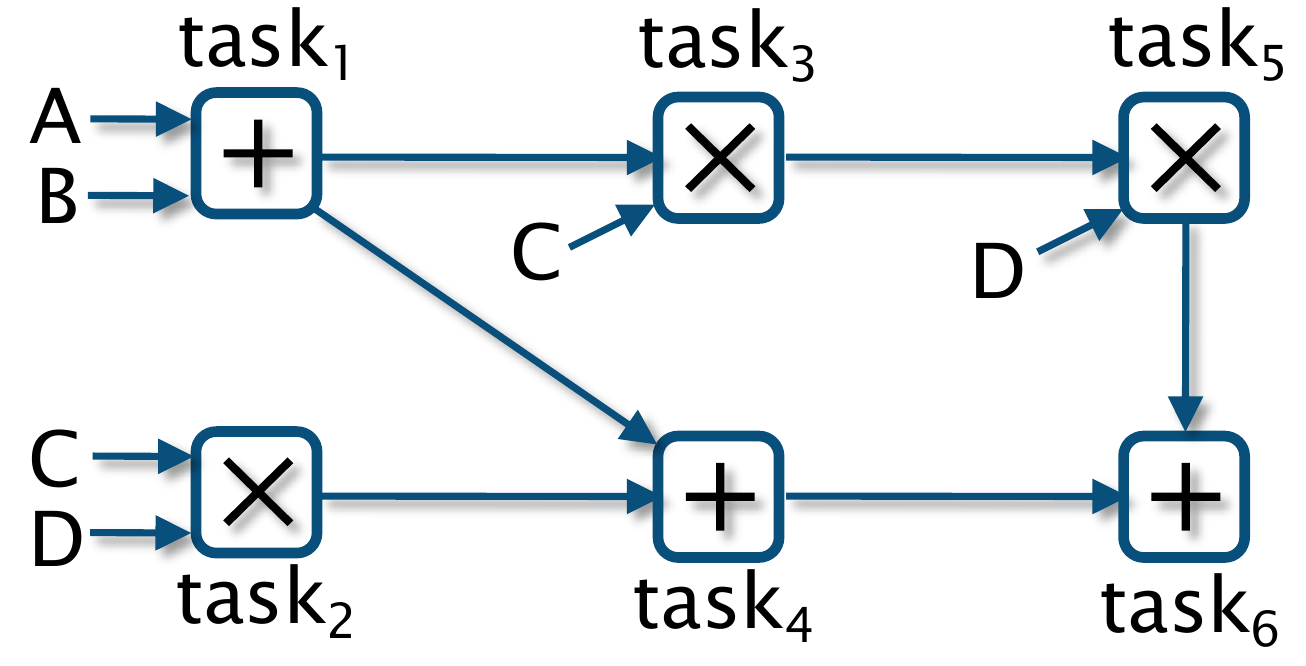}
\vspace*{-0.4cm}
\caption{Task graph for computing $D {\times} (C {\times} (A{+}B)) {+} ((A{+}B){+}(C{\times}D))$}\label{task-pta-fig}
\vspace*{-0.4cm}
\end{figure}

\startpara{Processor Task Scheduling} This case study is based on the task-graph scheduling problem from~\cite{BFLM11}. The task-graph is given in \figref{task-pta-fig} and is for evaluating the expression $D {\times} (C {\times} (A{+}B)) {+} ((A{+}B){+}(C{\times}D))$ where each multiplication and addition is evaluated on one of two processors, $P_1$ and $P_2$. The time and energy required to perform these operations is different, with $P_1$ being faster than $P_2$ while consuming more energy as detailed below.
\begin{itemize}
\item Time and energy usage of $P_1$: $[0,2]$ picoseconds for addition, $[0,3]$ picoseconds multiplication, 10 Watts when idle and 90 Watts when active.
\item Time and energy usage of $P_2$: $[0,5]$ picoseconds for addition, $[0,7]$ picoseconds multiplication, 20 Watts when idle and 30 Watts when active.
\end{itemize}
A (non-probabilistic) TA model is considered in~\cite{BFLM11}, which is the parallel composition of a TA for each processor and for the scheduler. 
Previously, in~\cite{NPS13}, we extended this model by adding probabilistic behaviour to give a PTA. However, the execution time of the processors had to remain fixed since the non-determinism was under the control of the scheduler, and therefore the optimal scheduler would always choose the minimum execution time for each operation. By moving to a TPTG model, we can allow the execution times to be under the control of a separate player (the environment). We further extend the model by allowing the processors $P_1$ and $P_2$ to have at most $k_1$ and $k_2$ faults respectively. We assume that the probability of any fault causing a failure is $p$ and faults can happen at any time a processor is active, i.e., the time the faults occur is under the control of the environment. Again, we could not model this extension with a PTA, since the scheduler would then be in control of when the faults occurred, and therefore could decide that no faults would occur.

As explained in \cite{BFLM11}, an optimal schedule for a game model in which delays can vary does not yield a simple assignment of tasks to processors at specific times as presented in \cite{NPS13} for PTAs, but instead it is an assignment that also has as input when previous tasks were completed and on which processors.

\begin{figure}[t]
\centering
{\subfigure[Processor $P_1$]{\includegraphics[scale=0.26]{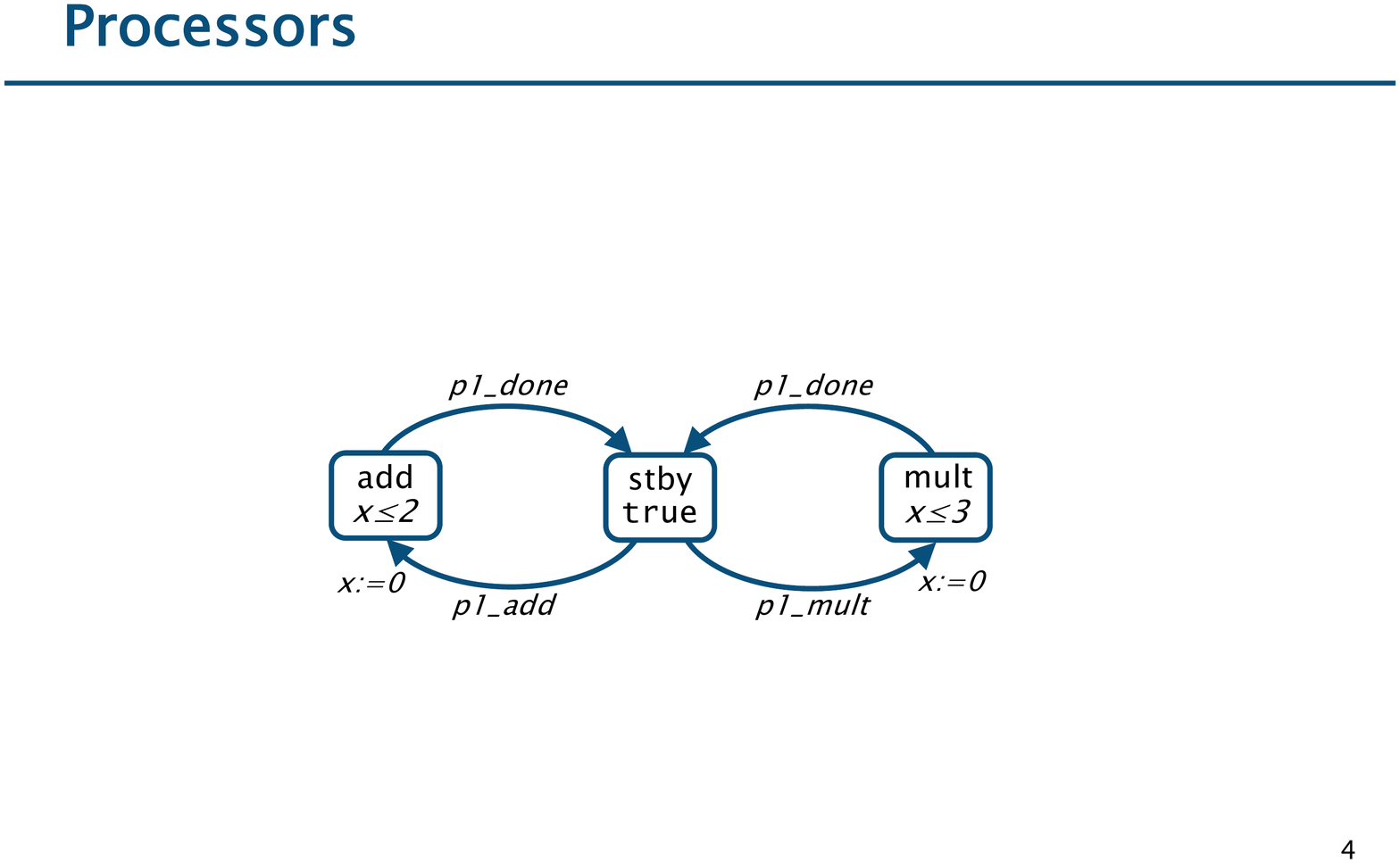}}}
\hfil
{\subfigure[Faulty version of processor $P_1$]{\includegraphics[scale=0.26]{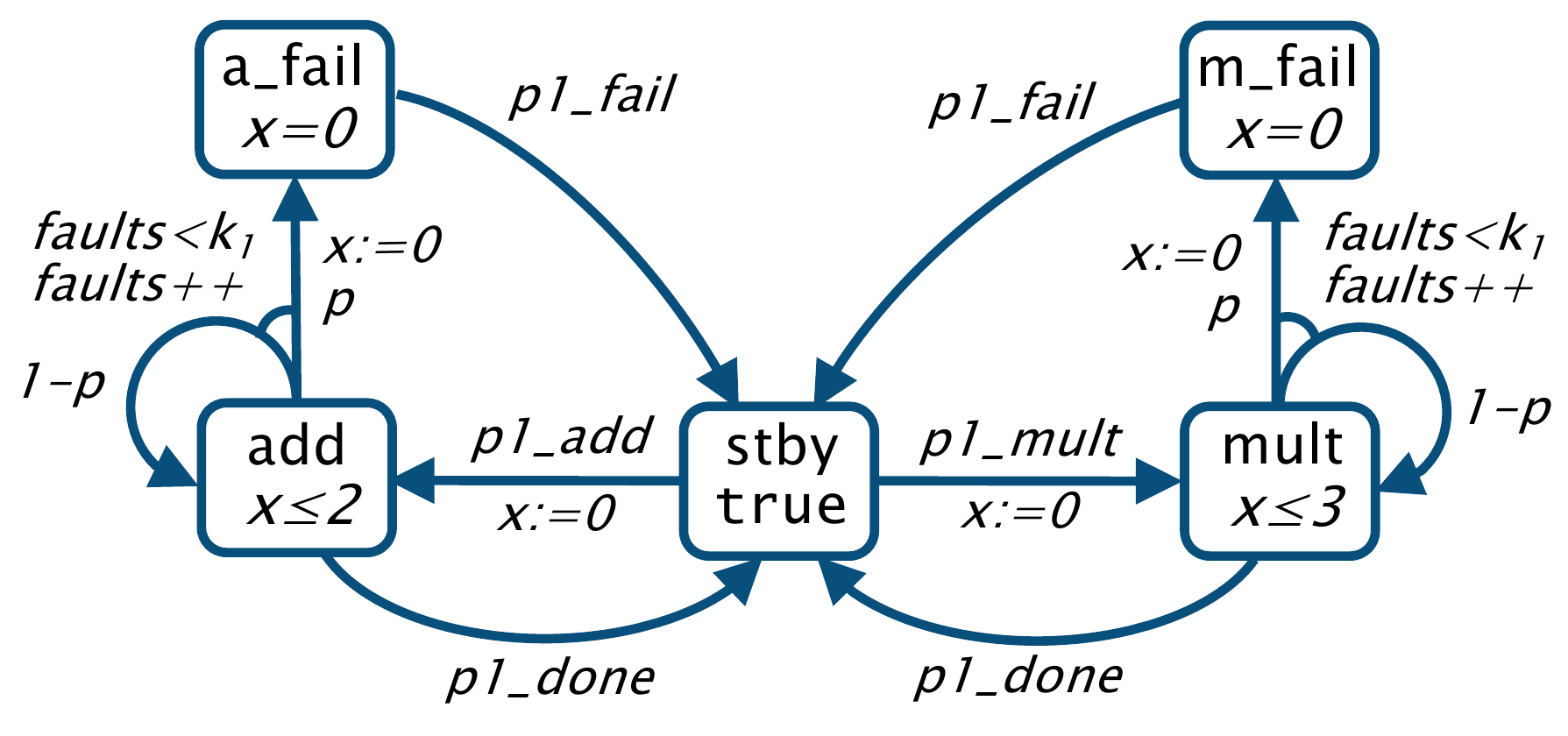}}}
\vspace*{-0.4cm}
\caption{PTAs for the task-graph scheduling case study}\label{processors-fig}
\vspace*{-0.6cm}
\end{figure}

In \figref{processors-fig} we present both the original TA model for processor $P_1$, in which the execution time is non-deterministic, and the extended PTA, which allows $k_1$ faults and where the probability of a fault causing a failure equals $p$. The PTA includes an integer variable $\mathit{faults}$ and the missing enabling conditions equal $\mathtt{true}$. To specify the automaton for the scheduler and ensure that we can then build a turn-based game, we restrict the scheduler so that it decides what tasks to schedule initially and immediately after a task ends, then passes control to the environment, which decides the time for the next active task to end.

\begin{figure}[t]
\centering
\vspace*{-0.1cm}
\includegraphics[scale=0.25]{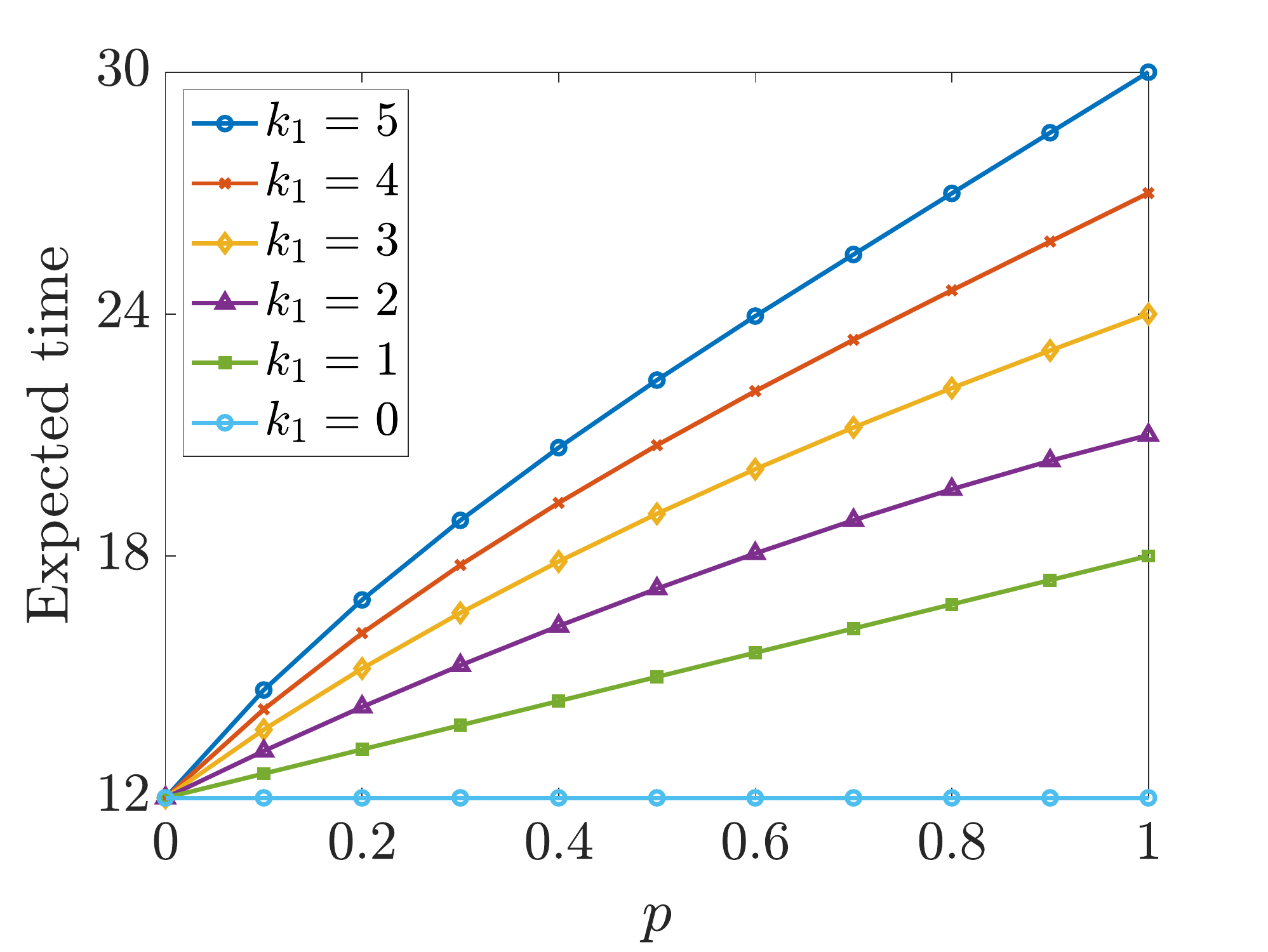}
\hfil
\includegraphics[scale=0.25]{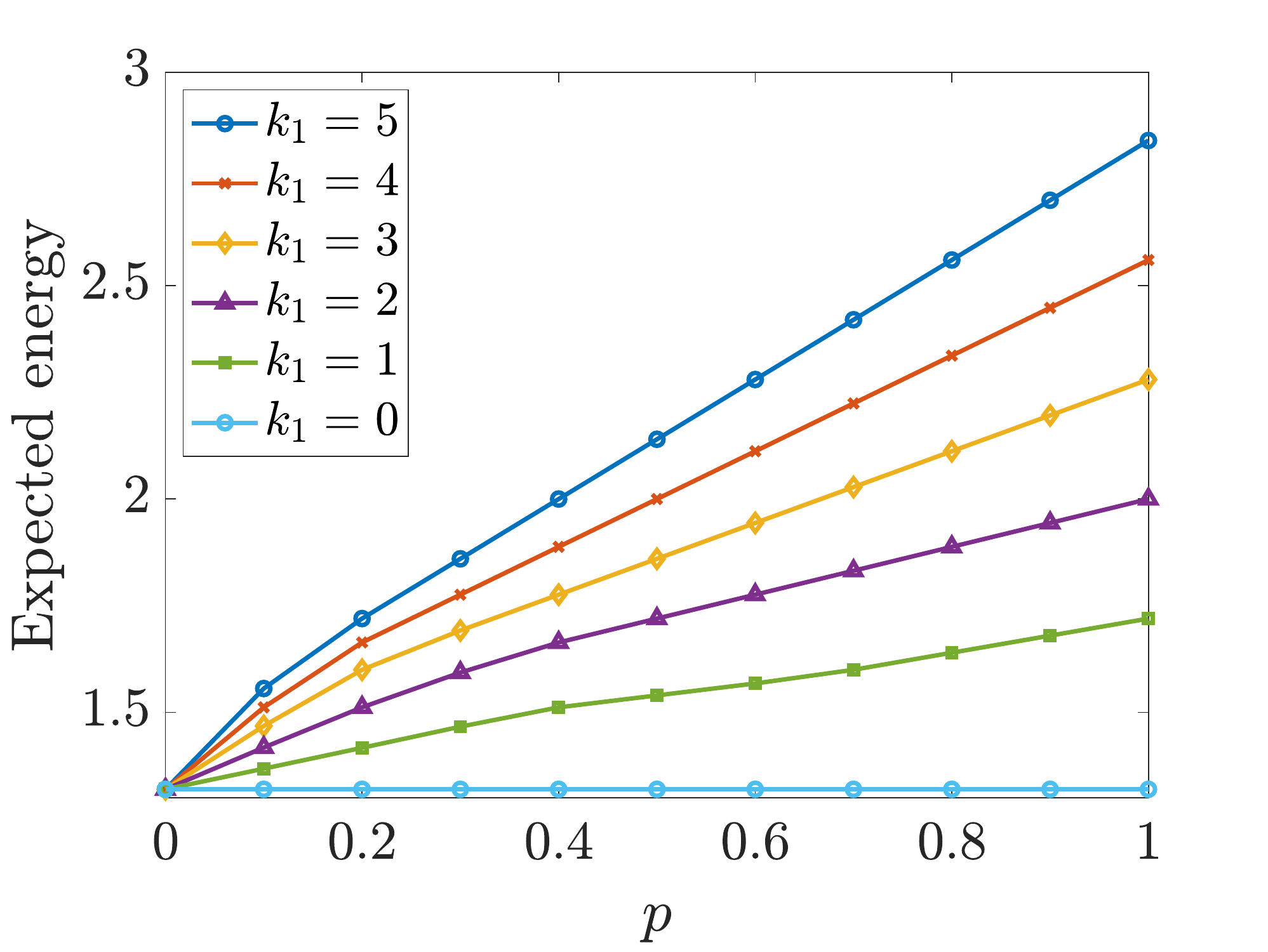}
\vspace*{-0.4cm}
\caption{Minimum expected time and energy to complete all tasks ($k_2{=}k_1$)}\label{task-expected-fig}
\vspace*{-0.6cm}
\end{figure}

In \figref{task-expected-fig} we have plotted both the optimal expected time and energy when there are different number of faults in each processor as the parameter $p$ (the probability of a fault causing a failure) varies. As would be expected, both the optimal expected time and energy consumption increases both as the number of faults increases and the probability that a fault causes a failure increases.

Considering the synthesised optimal schedulers for the expected time case, when $k_1{=}k_2{=}1$ and $p{=}1$, the optimal approach is to just use the faster processor $P_1$ and the expected time equals $18.0$. The optimal strategy for the environment, i.e., the choices that yield the worst-case expected time, against this scheduler is to delay all tasks as long as possible and cause a fault when a multiplication task is just about to complete on $P_1$ (recall $P_2$ is never used under the optimal scheduler). A multiplication is chosen as this takes longer (3 picoseconds) than an addition task (2 picoseconds). These choices can be seen through the fact that $18.0$ is the time for 4 multiplications and 3 additions to be performed on $P_1$, while the problem requires 3 multiplications and 3 additions. As soon as the probability of a fault causing a failure is less than 1, the optimal scheduler does use processor $P_2$ from the beginning by initially scheduling $\mathit{task}_1$ on process $P_1$ and $\mathit{task}_2$ on processor $P_2$ (which is also optimal when no faults can occur). 

In the case of the expected energy consumption, the optimal scheduler uses both processes unless one has 2 or more faults than the other and there is only a small chance that a fault will cause a failure. For example, if $P_1$ has 3 faults and $P_2$ has 1 fault, then $P_1$ is only used by the optimal scheduler when the probability of a failure causing a fault is approximately $0.25$ or less.

\section{Conclusions}

We have introduced turn-based probabilistic timed games and shown that digital clocks are sufficient for analysing a large class of such games and performance properties. We have demonstrated the feasibility of the approach through two case studies. However, there are limitations of the method since, in particular, as for PTAs~\cite{KNPS06}, the digital clocks semantics does not preserve stopwatch properties or general (nested) temporal logic specifications.

We are investigating extending the approach to concurrent probabilistic timed games. However, since such games are not determined for expected reachability properties~\cite{FKNT16}, this is not straightforward. One direction is to find a class of %
games which are determined. If we are able to find such a class, then the extension of PRISM-games to concurrent stochastic games~\cite{KNPS18} could be used to verify this class. Work on finite-state concurrent stochastic games has recently been extended to the case when players have distinct objectives~\cite{KNPS19} and considering such objectives in the real-time case is also a direction of future research.

Another direction of future research is to  %
formulate a zone-based approach for verifying probabilistic timed games. For the case of probabilistic reachability, this appears possible through the approach of~\cite{KNP09c} for PTAs. However, it is less clear that the techniques for expected time~\cite{JKNQ17} and expected prices~\cite{KNP17b}, and temporal logic specifications~\cite{KNSW07} for PTAs, can be extended to TPTGs. Finally, we mention that, although the PRISM language models used in \sectref{case-sect} were built by hand, in future we plan to automate this procedure, extending the one already implemented in PRISM~\cite{KNP11} for PTAs.

\startpara{Acknowledgements} This work is partially supported by the EPSRC Programme Grant on Mobile Autonomy and the PRINCESS project, under the DARPA BRASS programme  (contract FA8750-16-C-0045).

\bibliographystyle{spmpsci}
\bibliography{bib}

\ifthenelse{\isundefined{\techreport}}{%
}{%
\appendix
\section{Proofs from \sectref{correct-sect}}\label{appendix}

In this appendix we include the details omitted from \sectref{correct-sect} as they closely follow those for PTAs presented in~\cite{KNPS06}. As in \sectref{correct-sect}, we fix a TPTG $\ptg$, coalition of players $C$ and target set of locations $F \in \loc$, and let $F_\Tset = \{ (l,v) \in F {\times} \Tset^\clocks \mid v \sat \inv(l) \}$ be the target of game $\tptg$ for $\Tset \in \{ \Rset,\Nset\}$. 

We first prove \propref{strat2-prop}, i.e.\ that for any strategy profile of $\nptgc$ there is a corresponding strategy profile of $\rptgc$.
\begin{proof}[of \propref{strat2-prop}]
Consider any strategy profile $\sigma'{=}(\sigma_1',\sigma_2')$ of $\nptgc$. We can construction the strategy profile $\sigma{=}(\sigma_1,\sigma_2)$ of $\rptgc$ where the strategies $\sigma_1$ and $\sigma_2$ make the same choices as those of $\sigma_1'$ and $\sigma_2'$ respectively. The only difference is in the states reached as values of a clocks in $\rptg$ are not bounded. It then follows that $\Pset^{\sigma}(F_\Rset) = \Pset^{\sigma'}(F_\Nset)$ and $\Eset^{\sigma}(F_\Rset) = \Eset^{\sigma'}(F_\Nset)$ as required. \qed
\end{proof}
We now presenting definitions and results for PTAs~\cite{NPS13} and TAs~\cite{Hen91,HMP92} used in proving the correctness of the digital clocks semantics.
\begin{definition}\label{digi-def}
For any $t \in \Rset$ and $\varepsilon \in [0,1]$ let:
\[
\digi{t}{\varepsilon} =
\left\{ \begin{array}{cl}
\lfloor t \rfloor & \mbox{if $t \leq \lfloor t \rfloor + \varepsilon$} \\ 
\lceil t \rceil & \mbox{otherwise} .
\end{array}
\right.
\]
\end{definition}
\begin{lemma}\label{digi-lem}
For any $t,t' \in \Rset$, $c \in \Nset$ and $\sim \in \{ \leq , = , \geq \}$,
if $t - t' \sim c$ then $\digi{t}{\varepsilon} - \digi{t'}{\varepsilon} \sim c$ for all $\varepsilon \in [0,1]$.
\end{lemma}
\begin{definition}
For any infinite path $s_0 \xrightarrow{t_0,a_0} s_1  \xrightarrow{t_1,a_1} \cdots$ of $\tpta$, the accumulated duration up to the $(n{+}1)$th state of $\pi$ is defined by $\dur{\pi,n} = \sum_{i=0}^{n-1} t_i$.
\end{definition}
\begin{lemma}\label{path-dig}
For any path $\pi= ( l_0 , v_0  ) \xrightarrow{t_0,a_0}  ( l_1 , v_1  )  \xrightarrow{t_1,a_1} \cdots$ of\/ $\rptgc$, 
$x \in \cX$ and $i \in \Nset$, there exists $j \leq i$ such that
$v_i(x) = \dur{\pi,i} - \dur{\pi,j}$.
\end{lemma}
\begin{definition}\label{pathdig-def}
For any (finite or infinite) path
$\pi = ( \bar{l} , \mathbf{0}  ) \xrightarrow{t_0,a_0}  ( l_1 , v_1  )  \xrightarrow{t_1,a_1} \cdots$
of\/ $\rptgc$, its $\varepsilon$-digitization is the path 
\[ \begin{array}{c}
\digi{\pi}{\varepsilon} = ( \bar{l} , \mathbf{0}  ) \xrightarrow{t_1',a_0} ( l_1 , \digi{v_1}{\varepsilon} ) \xrightarrow{t_1',a_1} \cdots
\end{array} \]
of\/ $\nptgc$ where for any $i \in \Nset$ and $x \in \cX:$
\begin{itemize}
\item
$\digi{v_i}{\varepsilon}(x) = \min(\digi{\dur{\pi,i}}{\varepsilon} - \digi{\dur{\pi,j}}{\varepsilon}, k_x+1)$
and $j \leq i$ such that $v_i(x) = \dur{\pi,i} - \dur{\pi,j}$ which exists by \lemref{path-dig};
\item
$t_i' =  \digi{\dur{\pi,i{+}1}}{\varepsilon} - \digi{\dur{\pi,i}}{\varepsilon}$.
\end{itemize}
\end{definition} 
The correctness of this construction is dependent on the clock constraints appearing in $\ptg$ being closed and diagonal-free (\assumref{pta-assum}). 

We next extend the notion of digitization to strategy profiles of $\rptgc$. Following \cite{KNPS06}, to achieve this we define the digitization of sets of paths, i.e.\ for a set of paths $\Pi$, we let $[\Pi]_\varepsilon \rmdef \{ [\pi]_\varepsilon \mid \pi \in \Pi \}$. When considering a strategy profile $\sigma$ of $\rptgc$, we assume the domain 
of the mapping $\digi{\cdot}{\varepsilon}$ is restricted to the sets of paths $\ipaths^{\sigma}$ and $\fpaths^{\sigma}$. Furthermore, for any set of finite paths $\Pi \subseteq \fpaths^{\sigma}$ let:
\[
\Prob^{\sigma}(\Pi) \; = \; \Prob^{\sigma}(\{ \pi' \in \ipaths^{\sigma} \mid \exists \pi\in \Pi . \, (\mbox{$\pi$ is a prefix of $\pi'$}) \})
\]
and for any finite path $\pi \in \fpaths^{\sigma}$, $1{\leq}i{\leq}2$ and $(t,a) \in A(\last(\pi))$ let:
\[ 
\pi \xrightarrow{t,a}  \; \rmdef \; \{ \pi' \in \fpaths^{\sigma} \mid \exists s \in S. \, ( \pi' = \pi \xrightarrow{t,a} s ) \} \, .
\]
\begin{definition}\label{estrat-def}
For any strategy profile $\sigma {=} (\sigma_1,\sigma_2)$ of $\rptgc$ and $\varepsilon \in [0,1]$ its $\varepsilon$-digitization strategy profile $\sigma^\varepsilon {=} (\sigma_1^\varepsilon,\sigma_2^\varepsilon)$ of $\nptgc$ is defined as follows. For any finite path $\pi$ of $\nptgc$ and $1{\leq}i{\leq}2$ such that $\last(\pi) \in S_i:$
\begin{itemize}
\item
if $\digi{\pi}{\varepsilon}^{-1}$ is non-empty, then for any $(t,a) \in A(\last(\pi))$, then the probability of $\sigma_i^\varepsilon$ choosing $(t,a)$ after $\pi$ has been performed is given by
\[
\sigma_i^\varepsilon(\pi)(t,a) \; = \;
\frac{\displaystyle \Prob^{\sigma}(\digi{\pi \xrightarrow{t,a}\;}{\varepsilon}^{-1})}
{\displaystyle \Prob^{\sigma}(\digi{\pi}{\varepsilon}^{-1})} 
\]
\item
if $\digi{\pi}{\varepsilon}^{-1}$ is the empty set, then let $\sigma_i^\varepsilon$ choose an arbitrary element of $A(\last(\pi))$.
\end{itemize}
\end{definition}
We next show how probability measures of strategy profiles change under $\varepsilon$-digitization.
\begin{proposition}\label{strat1-prop} 
For any strategy profile $\sigma$ of $\rptgc$ and $\varepsilon \in [0,1]$, we have
$\Prob^{\sigma^\varepsilon}(\Omega)=\Prob^{\sigma}(\digi{\Omega}{\varepsilon}^{-1})$
for all elements $\Omega$ of the $\sigma$-algebra of the probability measure $\Prob^{\sigma^\varepsilon}$.
\end{proposition}
\begin{proof}
Consider any strategy profile $\sigma{=}(\sigma_1,\sigma_2)$ of $\rptgc$ and $\varepsilon \in [0,1]$. Using \defref{estrat-def} we can build the strategy profile $\sigma^\varepsilon{=}(\sigma^\varepsilon_1,\sigma^\varepsilon_2)$ of $\nptgc$.
Now, from the construction of $\Prob^{\sigma^\varepsilon}$, see~\cite{KSK76}, it is sufficient to show that:
\begin{equation}\label{digi1-eqn}
\Prob^{\sigma^\varepsilon}(\pi) = \Prob^{\sigma} (\digi{\pi}{\varepsilon}^{-1}) \; \; \mbox{for all $\pi \in \fpaths^{\sigma^\varepsilon}$}
\end{equation}
which we prove by induction on the length of $\pi$. Therefore, consider any path $\pi \in \ipaths^{\sigma^\varepsilon}$. If $|\pi|{=}0$, then $\pi{=}\sinit$ and
$\Prob^{\sigma^\varepsilon}(\pi)=1=\Prob^{\sigma}(\digi{\pi}{\varepsilon}^{-1})$ as required.

Next, suppose $|\pi|{=}n{+}1$ and by induction the lemma holds for all paths of length $n$. Now, $\pi$ is of the form $\pi' \xrightarrow{t,a} s'$ for some path $\pi'$
of length $n$, $(t,a) \in A(\last(\pi'))$ and $s' \in S$. Now $\last(\pi') \in S_i$ for some $1{\leq}i{\leq}2$, and therefore by construction of $\Prob^{\sigma^\varepsilon}$ supposing $\delta(\last(\pi'),(t,a))=\mu$:
\begin{align*}
\lefteqn{\Prob^{\sigma^\varepsilon}(\pi) \; = \; \Prob^{\sigma^\varepsilon}(\pi') \cdot \sigma_i^{\varepsilon}(\pi')(a) \cdot \mu(s')}  \\
& = \; \Prob^{\sigma} (\digi{\pi'}{\varepsilon}^{-1}) \cdot \sigma_i^{\varepsilon}(\pi')(t,a) \cdot \mu(s') &  \mbox{by induction} \\
& = \; \Prob^{\sigma}(\digi{\pi'}{\varepsilon}^{-1}) \cdot 
\frac{\displaystyle \Prob^{\sigma}(\digi{\pi' \xrightarrow{t,a}\;}{\varepsilon}^{-1})}{\displaystyle \Prob^{\sigma}(\digi{\pi'}{\varepsilon}^{-1})} \cdot \mu(s') & \qquad \mbox{by \defref{estrat-def}} \\
& = \; \Prob^{\sigma}(\digi{\pi' \xrightarrow{t,a}\;}{\varepsilon}^{-1}) \cdot \mu(s') & \mbox{rearranging} \\
& = \; \Prob^{\sigma} \{ \pi \in \fpaths^{\sigma_1,\sigma_2} \mid  \digi{\pi}{\varepsilon}= \pi' \xrightarrow{t,a} s' \}  & \!\! \mbox{by definition of $\Prob^{\sigma}$} \\ 
& = \; \Prob^{\sigma}(\digi{\pi}{\varepsilon}^{-1}) & \mbox{by construction of $\pi$} 
\end{align*}
and hence \eqnref{digi1-eqn} holds by induction. \qed
\end{proof}
We are now in a position to prove \thmref{reach-thm}, i.e.\  $\Pset_\rptg^C (F_{\Rset})  = \Pset_\nptg^C (F_{\Nset} )$.
\begin{proof}[of \thmref{reach-thm}] 
Consider any TPTG $\ptg$ satisfying \assumref{pta-assum}, coalition of players $C$ and set of locations $F \subseteq \loc$. From \defref{pathdig-def}, we have that:
\begin{equation}\label{path-eqn}
\mbox{$\pi(i) \in F_{\Rset}$ if and only if $\digi{\pi}{\varepsilon}(i) \in F_{\Nset}$ for any path $\pi$ of $\rptgc$ and $i \in \Nset$.}
\end{equation}
Now for any strategy profile $\sigma{=}(\sigma_1,\sigma_2)$ of $\rptgc$, by definition of $\Pset^\sigma(F_{\Rset})$:
\begin{align}
\Pset^{\sigma}(F_{\Rset}) & \; = \; \Prob^{\sigma} (\{ \pi \in \ipaths^\sigma \mid \pi(i) \in F \; \mbox{for some} \; i \in \Nset \}) \nonumber \\
& \; = \; \Prob^{\sigma'} ([\{ \pi \in \ipaths^\sigma \mid \pi(i) \in F \; \mbox{for some} \; i \in \Nset \}]^{-1}_\varepsilon) \label{prob-eqn}
\end{align}
for some strategy profile $\sigma'{=}(\sigma_1',\sigma_2')$ of $\nptgc$
using \propref{strat1-prop}. Combining \eqnref{prob-eqn} with \eqnref{path-eqn} it follows that:
\begin{align*}
\Pset^{\sigma}(F_{\Rset}) \; & \lefteqn{\; = \; \Prob^{\sigma} \{ \pi \in \ipaths^{\sigma'} \mid \pi(i) \in F \; \mbox{for some} \; i \in \Nset \}} \\
&= \; \Pset^{\sigma'}(F_{\Nset}) & \mbox{by definition of $\Pset^{\sigma'}(F_{\Nset})$.}
\end{align*}
Since the player 2 strategy $\sigma_2$ was arbitrary it follows that:
\[
\inf\nolimits_{\sigma_2 \in \Sigma^2_{\rptgc}} \Pset^{\sigma_1,\sigma_2}(F_{\Rset})
\; \geq \;
\inf\nolimits_{\sigma_2' \in \Sigma^2_{\nptgc}} \Pset^{\sigma_1',\sigma_2'}(F_{\Nset}) \, .
\]
Following dual arguments and using \propref{strat2-prop} we have:
\[
\inf\nolimits_{\sigma_2 \in \Sigma^2_{\rptgc}} \Pset^{\sigma_1,\sigma_2}(F_{\Rset})
\; \leq \;
\inf\nolimits_{\sigma_2' \in \Sigma^2_{\nptgc}} \Pset^{\sigma_1',\sigma_2'}(F_{\Nset})
\]
and therefore:
\[
\inf\nolimits_{\sigma_2 \in \Sigma^2_{\rptgc}} \Pset^{\sigma_1,\sigma_2}(F_{\Rset})
\; = \;
\inf\nolimits_{\sigma_2' \in \Sigma^2_{\nptgc}} \Pset^{\sigma_1',\sigma_2'}(F_{\Nset}) \,  .
\]
We can now apply similar arguments again to yield:
\[
\sup\nolimits_{\sigma_1 \in \Sigma^1_{\rptgc}} \inf\nolimits_{\sigma_2 \in \Sigma^2_{\rptgc}} \Pset^{\sigma_1,\sigma_2}(F_{\Rset})
\; = \;
\sup\nolimits_{\sigma_1' \in \Sigma^1_{\nptgc}} \inf\nolimits_{\sigma_2' \in \Sigma^2_{\nptgc}} \Pset^{\sigma_1',\sigma_2'}(F_{\Nset})
\]
as required. \qed
\end{proof}
To give the proof of \lemref{variable-lem}, we first require the following results from linear programming.
\begin{definition} 
A matrix $\mathbf{A}$ is\/ {\em totally unimodular} if each subdeterminant of $\mathbf{A}$ is $0$, $+1$ or $-1$.
\end{definition}
\begin{theorem}[\cite{Sch86} Theorem 19.3]\label{lp1-thm}
Let $\mathbf{A}$ be a matrix with entries $0$, $+1$ and $-1$. Then the following are equivalent:
\begin{enumerate}
\item $\mathbf{A}$ is totally unimodular;
\item
each collection of columns of $\mathbf{A}$ can be split into two parts so that the sum of the columns in one part minus the
sum of the columns in the other part is a vector with entries only $0$, $+1$ and $-1$.
\end{enumerate}
\end{theorem}
\begin{theorem}[\cite{Sch86} Corollary 19.1.a]\label{lp2-thm}
Let $\mathbf{A}$ be a totally unimodular matrix, and let $\mathbf{b}$ and $\mathbf{c}$ be integral vectors. Then both problems in the linear programming duality equation:
\[
\max \{ \mathbf{c}\mathbf{x} \mid \mathbf{x} {\geq} 0 \wedge \mathbf{A}\mathbf{x} {\leq} \mathbf{b} \} = \min \{ \mathbf{y}\mathbf{b} \mid \mathbf{y} {\geq} 0 \wedge \mathbf{y}\mathbf{A} {\geq} \mathbf{c} \}
\]
have integral optimum solutions.
\end{theorem}
We also require a variant of \defref{EB-def} which considers only action prices.
\begin{definition}\label{EB2-def}
For $\Tset \in \{ \Rset,\Nset\}$, strategy profile $\sigma{=}(\sigma_1,\sigma_2)$ of $\tptgc$ and  
finite path $\pi$ of the profile let $\Aset^{\sigma_1,\sigma_2}_{0} (\pi,F_{\Tset})=0$ and for any $n \in \Nset$, if $\mathit{last}(\pi){=}(l,v) \in S_i$ for $1{\leq}i{\leq}2$, $\sigma_i(\pi){=}(t,a)$ and $\mu = P_\tptg((l,v),(t,a))$, then:
\begin{align*}
\Aset^{\sigma}_{n+1} (\pi,F_{\Tset}) \; = \;  \left\{ \begin{array}{cl}
0 & \;\mbox{if $(l,v) \in F_{\Tset}$} \\
\arew(l,a) + \sum\limits_{s' \in S} \mu(s') 
\cdot \Aset^{\sigma}_{n} (\pi \xrightarrow{t,a}s',F_{\Tset}) & \; \mbox{otherwise.}
\end{array} \right.
\end{align*}
\end{definition}
We are now in a position to prove \lemref{variable-lem}, i.e. for any strategy profile $\sigma$ of $\rptgc$
and $n \in \Nset$, there exist strategy profiles $\sigma^{\mathit{lb}}$ and $\sigma^{\mathit{ub}}$ of $\nptgc$ such that:
\[
\Eset^{\sigma^{\mathit{lb}}}_{n}(F_{\Nset}) 
\; \leq \; 
\Eset^{\sigma}_{n}(F_{\Rset})
\; \leq \;
\Eset^{\sigma^{\mathit{ub}}}_{n}(F_{\Nset}) \, .
\]
\begin{proof}[of \lemref{variable-lem}]
Consider any strategy profile $\sigma{=}(\sigma_1,\sigma_2)$ of $\rptgc$.
The first part of the proof involves constructing a set of constraints on the time steps of the strategy profiles that follow the same action choices as $\sigma$ up until the $n$th step. Using these constraints we then formulate a linear programming problem, whose objective is either to maximise or minimise the expected price of reaching a set of target states within $n$ steps. The result then follows from showing that there exist integer solutions which achieve the maximum and minimum values. Below, we consider only the construction of the strategy profile $\sigma^{\mathit{ub}}{=}(\sigma_1^{\mathit{ub}},\sigma_2^{\mathit{ub}})$ as the construction of the profile $\sigma_1^{\mathit{lb}}$ follows similarly.

We begin by constructing a set of linear constraints from which we can derive a set of strategy profiles that make the same choices as $\sigma_1$. More precisely, we consider any sequence of real values $\mathbf{t} =\langle t_{\pi} \rangle_{\pi \in \fpaths^{\sigma}}$
which satisfy, for any $\pi \in \fpaths^{\sigma}$, the following constraints:  
\begin{itemize}
\item
if $|\pi|=0$, then
\begin{subequations}
\begin{eqnarray}
t_\pi & \; \geq \;  & 0  \label{constraints1-eqn} \\
-t_\pi& \; \geq \;  & 0
\end{eqnarray}
\item if $|\pi|>0$, then 
\begin{align}
t_\pi - t_{\pi^{(k)}} & \; \geq \;   \lfloor \dur{\pi,|\pi|}  -  \dur{\pi,k} \rfloor & \mbox{for all $k <|\pi|$} \\
- t_\pi + t_{\pi^{(k)}} & \; \geq \;     - \lceil \dur{\pi,|\pi|}  +  \dur{\pi,k} \rceil & \mbox{for all $k <|\pi|$} \label{constraints4-eqn}
\end{align}
\end{subequations}
\end{itemize}
Sequences of values which satisfy these constraints include letting $t_\pi$ equal $\dur{\pi,|\pi|}$
or $\digi{\dur{\pi,|\pi|}}{\varepsilon}$ for any $\varepsilon \in [0,1]$. 

Now suppose that we fix a sequence of real values $\mathbf{t}$ which satisfy the above constraints. From these values we can construct the choices of a strategy profile $(\sigma_1^{\mathbf{t}},\sigma_2^{\mathbf{t}})$ which match the action choices of the profile $\sigma$, although the durations may differ.
The set of finite paths of the profile $(\sigma_1^{\mathbf{t}},\sigma_2^{\mathbf{t}})$ when starting from the initial state is given by $\{ [\pi]_{\mathbf{t}} \mid \pi \in \fpaths^{\sigma_1,\sigma_2} \}$, which we define inductively as follows:
\begin{itemize}
\item
if $\pi=( \linit, \mathbf{0} )$, then $[\pi]_{\mathbf{t}} = \pi$;
\item
if $\pi$ is of the form $\pi' \xrightarrow{t,a} ( l,v )$,
then:
\[
[\pi]_{\mathbf{t}} \; \rmdef \; [\pi']_{\mathbf{t}} \xrightarrow{t',a} ( l, v' )
\]
where $t' =  t_\pi - t_{\pi'}$ and for any clock $x$ we have $v'(x) = t_{\pi} - t_{\pi^{(j)}}$ for $j \leq |\pi|$ such that $v(x) = \dur{\pi,|\pi|} - \dur{\pi,j}$, which exists by \lemref{path-dig}.
\end{itemize}
The construction all finite paths of the profile $(\sigma_1^{\mathbf{t}},\sigma_2^{\mathbf{t}})$ yields also the choices made by the profile and by construction the action choices match those of the profile $\sigma$ although the durations may differ. The fact that these choices are valid choices of $\rptgc$ follows from \lemref{path-dig}, equations \eqnref{constraints1-eqn}--\eqnref{constraints4-eqn}
and since we restrict attention to closed, diagonal-free probabilistic timed games (\assumref{pta-assum}).

For any state $(l,v) \in S$, to simplify the presentation let $\lrew(l,v)=\lrew(l)$.
Now, from the construction of the choices of $(\sigma_1^{\mathbf{t}},\sigma_2^{\mathbf{t}})$ and \defref{EB-def} and \defref{EB2-def}, it follows that $\Eset^{\sigma_1^{\mathbf{t}},\sigma_2^{\mathbf{t}}}_{n}(F_{\Rset})$ equals:
\begin{eqnarray}
\lefteqn{\hspace*{-2cm} \sum_{\substack{\pi \in \fpaths^{\sigma} \wedge |\pi|\leq n \\ \wedge\forall i \leq |\pi| . \, \pi(i) \not\in F_{\Rset}}}  \!\!\!\! \!\!\!\!\!\!\!\!
\Prob^{\sigma_1^{\mathbf{t}},\sigma_2^{\mathbf{t}}}(\pi) \cdot (t_\pi {-} t_{\pi^{|\pi|{-}1}}) \cdot \lrew(\last(\pi)) 
+ \Aset^{\sigma_1^{\mathbf{t}},\sigma_2^{\mathbf{t}}}_n(F_{\Rset})} \nonumber\\ 
\quad\quad &\!\!\!\!\!\!\!\!\!=&  
\!\!\!\!\!\!\!\!\!\!\sum_{\substack{\pi \in \fpaths^{\sigma} \wedge |\pi|\leq n \\ \wedge\forall i \leq |\pi| . \, \pi(i) \not\in F_{\Rset}}}   \!\!\!\! \!\!\!\!\!\!\!\!
\Prob^{\sigma}(\pi) \cdot  (t_\pi {-} t_{\pi^{|\pi|{-}1}}) \cdot \lrew(\last(\pi))
+ \Aset^{\sigma}_n(F_{\Rset}) \label{obj-eqn}
\end{eqnarray}
since the action choices of $\sigma$ and $(\sigma_1^{\mathbf{t}},\sigma_2^{\mathbf{t}})$ are the same.

Now suppose we fix some
$n \in \Nset$ and consider the following linear programming problem over the variables $\langle t_\pi \rangle_{\pi \in \fpaths^{\sigma}_n}$, where
$\fpaths^{\sigma}_n$ is subset of paths of $\fpaths^{\sigma}$ with length at most $n$:
maximise \eqnref{obj-eqn} such that the constraints of \eqnref{constraints1-eqn}--\eqnref{constraints4-eqn} are satisfied. From \assumref{pta-assum} we have that all probabilities are rational, and therefore we can scale the objective function such that it contains only integer values. Furthermore, from the construction of the constraints the corresponding matrix is totally unimodular (using \thmref{lp1-thm} and the fact that, for any collection of columns of the constraint matrix, the sum of the columns is a vector with entries only $0$, $+1$ and $-1$). Therefore, using \thmref{lp2-thm}, it follows that the maximum solution is achieved by an integer vector.  More precisely, there exists a strategy profile $\sigma^{\mathit{ub}}{=}(\sigma_1^{\mathit{ub}},\sigma_2^{\mathit{ub}})$ of $\nptgc$ such that $\Eset^{\sigma}_{n}(F_{\Rset}) \leq \Eset^{\sigma^{\mathit{ub}}}_{n}(F_{\Nset})$
as required. \qed
\end{proof}

}

\end{document}